\documentclass[manuscript]{acmart}
\usepackage{graphicx}
\usepackage{textcomp}
\usepackage{comment}

\usepackage{filecontents}

\newtheorem{theorem}{Theorem}
\newtheorem{lemma}{Lemma}
\newtheorem{corollary}{Corollary}
\newtheorem{definition}{Definition}
\newtheorem{proposition}{Proposition}
\usepackage{algorithmic}
\usepackage{braket}
\usepackage{booktabs, array, tabularx, makecell}
\usepackage{xcolor}
\usepackage[T1]{fontenc}
\usepackage{enumerate}
\usepackage{enumitem}
\usepackage{subcaption}
\usepackage{units}
\usepackage{multirow}
\usepackage[normalem]{ulem}
\usepackage{mathtools}
\newcolumntype{P}[1]{>{\centering\arraybackslash}p{#1}}
\newcolumntype{M}[1]{>{\centering\arraybackslash}m{#1}} 

\AtBeginDocument{%
  }

\begin{document}

\title{A Risk-Aware Framework for Covert Quantum Communication under Stochastic Channel Uncertainty}

\author{Abbas Arghavani}
\email{abbas.arghavani@mdu.se}
\orcid{0000-0001-6581-2251}
\affiliation{%
  \institution{Department of Computer Science and Engineering, M{\"a}lardalen University}
  \city{V{\"a}ster{\aa}s}
 \country{Sweden}
}

\author{Shahid Raza}
\affiliation{%
  \institution{School of Computing Science, University of Glasgow}
  \city{Glasgow}
  \country{Scotland, UK}}
\email{shahid.raza@glasgow.ac.uk}

\author{Maryam Amiri}
\affiliation{%
  \institution{Department of Physics, University of Otago}
  \city{Dunedin}
  \country{New Zealand}
}

\author{Alessandro Papadopoulos}
\affiliation{%
 \institution{Department of Computer Science and Engineering, M{\"a}lardalen University}
 \country{Sweden}}

\renewcommand{\shortauthors}{Arghavani et al.}

\begin{abstract}

Covert quantum communication (CQC) enables privacy-preserving information transmission by concealing not only message content but also the existence of communication. Existing CQC formulations typically assume deterministic or worst-case channel conditions, an assumption that is difficult to justify in realistic free-space optical and quantum communication environments affected by turbulence, background radiance fluctuations, and stochastic detector noise. This paper introduces a stochastic risk-aware optimization framework for CQC under uncertain physical-layer conditions. By modeling transmissivity and background noise as random variables, we formulate covertness and reliability guarantees probabilistically through chance-constrained optimization, with explicit outage budgets 
$\epsilon_{\text{cov}}$ and $\epsilon_{\text{rel}}$. This reframes CQC design as a risk-calibrated resource allocation problem balancing throughput, covertness, reliability, and communication privacy.
We derive tractable quantile-based reformulations for covertness and reliability outage constraints and characterize feasible operating regions under stochastic uncertainty. We also introduce a complementary risk-adjusted utility formulation to expose trade-offs between throughput maximization and probabilistic security or reliability violations. The analysis reveals a key operational transition: modest relaxations in acceptable covertness outage risk $\epsilon_{\text{cov}}$ can yield substantial gains in covert throughput, while aggressive throughput optimization can destabilize covertness outside sparse-transmission regimes.
Our Monte Carlo evaluation under log-normal fading and stochastic thermal noise demonstrates that the framework significantly enlarges feasible operating regions, improves covert throughput by more than an order of magnitude, and identifies critical degradation boundaries beyond which covert operation becomes unreliable.
By combining stochastic optimization with probabilistic security and privacy guarantees, this work advances CQC toward operationally realistic secure quantum networking for free-space, satellite, and low-probability-of-detection communications.

\end{abstract}




\maketitle

\section{Introduction}
\label{sec:introduction}

Modern cryptography protects the \emph{content} of communications, but not their \emph{existence}. In many high-stakes settings, merely detecting a signal can trigger censorship, jamming, or physical intervention, even if the message remains undeciphered. This motivates \emph{covert communication} (also known as low-probability-of-detection, LPD), whose goal is to hide that communication is happening at all.

Consider a canonical scenario: A legitimate transmitter (Alice) wishes to communicate with a legitimate receiver (Bob) while an adversarial warden (Willie) attempts to detect whether any transmission is taking place. A transmission is deemed \emph{covert} if Willie's optimal detector cannot perform significantly better than random guessing. Figure~\ref{fig:simple_model} depicts this setting in the optical domain.

\begin{figure}[!t]
    \centering
    \includegraphics[width=0.5\linewidth]{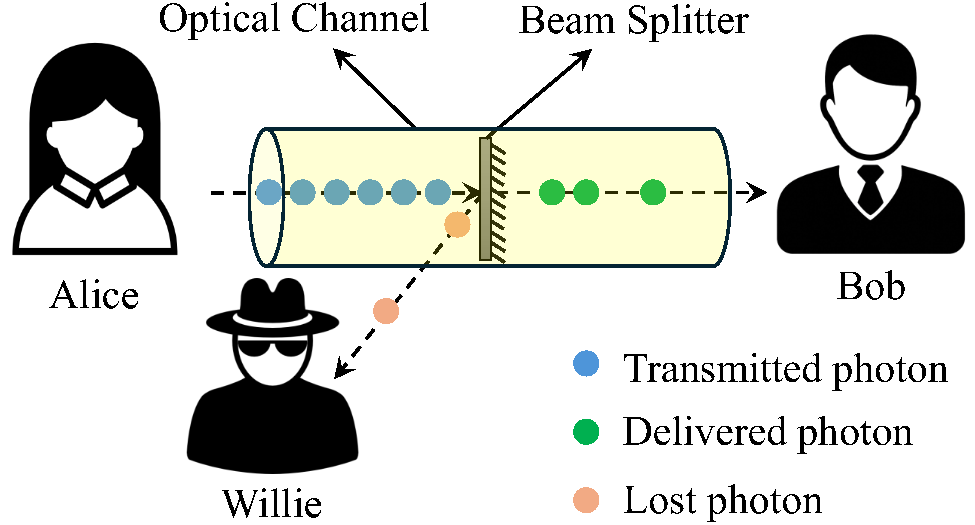}
    \caption{Schematic of CQC. Alice transmits photons to Bob over a channel with transmittance $\eta$. An adversary, Willie, collects the lost photons $(1-\eta)$ to detect the transmission, whose signal Alice must hide within background noise.}
    \label{fig:simple_model}
\end{figure}

Covert communication has been extensively studied in classical (non-quantum) wireless, wired, and optical systems. A foundational result, known as the \emph{Square-Root Law (SRL)}, states that over $n$ channel uses only $\mathcal{O}(\sqrt{n})$ bits can be sent covertly~\cite{bash2013limits, wang2016fundamental}. This limit has driven research on signaling, coding, and system design under strict undetectability constraints~\cite{arghavani2023covert}. In parallel, several works show that injecting or exploiting \emph{uncertainty} at the warden, e.g., via friendly jamming or fluctuating interference, can relax SRL assumptions and, under suitable models of channel/noise uncertainty, enable substantially higher covert throughput than the classical SRL would suggest~\cite{sobers2017covert, soltani2018covert, arghavani2023covert, arghavani2021game}.

Recent advances have extended covert-communication principles into the \emph{quantum} domain~\cite{anderson2024covert, anderson2024square, anderson2025achievability, tahmasbi2021signaling}. In covert quantum communication (CQC), Alice sends quantum states (typically single photons or weak coherent states) over a lossy optical channel to Bob, while Willie attempts to detect the mere existence of communication. Communication is considered covert if Willie cannot reliably infer whether a transmission is taking place. As sketched in Figure~\ref{fig:simple_model}, a fraction of the optical energy inevitably leaks into the environment and can be collected by Willie for detection; the physical and threat models are shown in Figure~\ref{fig:channel} and detailed in Section~\ref{sec:model}. Beyond its theoretical appeal, CQC enables emerging cybersecurity applications, including
\begin{itemize}
    \item \textbf{Satellite-to-ground links:} Undetectable command-and-control of space assets during sensitive missions.
    \item \textbf{Free-space optical (FSO) urban networks:} Invisible, secure backhaul in contested environments.
    \item \textbf{Quantum key distribution (QKD) integration:} Concealing the existence of the QKD session itself, not only the key.
    \item \textbf{Privacy-sensitive finance and tactical command and control:} Covert signaling for secure financial transactions and undetectable command-and-control in tactical settings.
\end{itemize}

Most existing CQC analyses assume that the legitimate parties know the channel's physical parameters, such as the transmittance and mean thermal photon number, \emph{exactly}, and that these parameters remain constant over time~\cite{anderson2024covert, anderson2024square, anderson2025achievability}. In practice, this assumption rarely holds. Real-world channels, particularly in FSO and satellite links, are subject to dynamic and unpredictable variations caused by atmospheric turbulence, pointing errors, and fluctuating background illumination. Such variability can significantly degrade performance or, worse, compromise covertness.

From a privacy-and-security perspective, the quantity at stake is not only throughput but the \emph{probability that the very existence of communication is exposed}. Under channel uncertainty, a design calibrated for nominal conditions can violate its intended covertness target or decoding guarantee when the environment deviates from that nominal model. The problem is therefore fundamentally one of \emph{security-risk calibration under uncertainty}: the system designer must choose an operating point that controls the probabilities of exposure and communication failure, not merely maximize rate under idealized assumptions.

A recent line of work has begun to address this issue through \emph{bounded uncertainty}, in which the channel parameters are assumed to lie within known deterministic intervals~\cite{arghavani2026robust}. That approach provides worst-case guarantees, but it is inherently conservative: it optimizes against the most adverse channel realizations even when such realizations may be rare. The present paper moves beyond that robust worst-case perspective by modeling the uncertainty stochastically and optimizing throughput under explicit probabilistic risk budgets.


In this paper, we take the next step by modeling channel uncertainty \emph{stochastically}. We treat key channel parameters (transmittance and thermal noise) as random variables drawn from physically motivated distributions (e.g., a truncated lognormal model for turbulence-induced loss, and a truncated Gaussian model for background noise). This shift enables a move from optimistic constant channel guarantees to a \emph{risk-aware} design paradigm in which system performance is optimized while explicitly bounding the probability of two failure events:
\begin{itemize}
    \item \textbf{Covertness failure}: The adversary's detection advantage exceeds the allowable security threshold, so the communication event is insufficiently hidden.
    \item \textbf{Decoding failure}: The channel is too degraded for the receiver to decode the message reliably.
\end{itemize}
%
To the best of our knowledge, prior CQC analyses have not jointly modeled physical-layer uncertainty through explicit random-variable channel models, imposed explicit probabilistic constraints on both covertness and reliability outages, and optimized throughput under these quantified risks. Our framework is intended to fill that gap under the adopted bosonic-channel model. The main contributions of this work are summarized as follows:
\begin{enumerate}[label=(\arabic*)]
    \item Develop a systematic framework for \textbf{risk-aware covert quantum communication} over channels with stochastic uncertainty, based on explicit probabilistic models for channel variation and outage risk.
    
    \item Formulate the \textbf{primary} design problem as a risk-constrained optimization that maximizes covert throughput under explicit probabilistic budgets on covertness and reliability outages. Theorem~\ref{thm:optimal_throughput} gives the optimal solution, and Corollary~\ref{cor:monotonicity} establishes monotonicity of the resulting Pareto frontier. As a \textbf{secondary exploratory extension}, we also study a risk-adjusted objective; Theorem~\ref{thm:foc_risk_adjusted} derives first-order conditions for any differentiable interior optimum, which we use only to interpret operating-point preferences rather than to claim hard covert guarantees.
    
    \item Show how stochastic physical-layer uncertainty can be reduced to calibrated covertness- and reliability-outage variables through the induced distributions of \(c_{\text{cov}}(\eta,\overline{n}_B)\) and \(R_{\text{ach}}(\eta,\overline{n}_B)\), and provide a practical, simulation-based methodology (using quantile reformulation and Monte Carlo estimation) to solve the resulting design problem and characterize the \textbf{risk–performance trade-off surface} and constrained Pareto behavior of CQC systems under the adopted stochastic model.
    
    \item Derive novel \textbf{analytical results}, including a closed-form solution for the optimal strategy in a benchmark exponential-noise channel (Lemma~1) and a formal analysis of risk sensitivity (Proposition~1) that identifies performance bottlenecks.    
    
    \item Demonstrate, through physically motivated FSO-inspired stochastic simulations, the usefulness of the framework and its ability to balance throughput and risk under realistic stochastic uncertainty.
\end{enumerate}

\noindent The remainder of this paper is organized as follows. Section~\ref{sec:related} reviews related work in covert and covert quantum communication. Section~\ref{sec:model} presents the system and threat models together with the baseline performance metrics. Section~\ref{sec:channel-models} introduces the stochastic channel model, operational assumptions, and risk definitions. Section~\ref{sec:problem} formulates the design problem, and Section~\ref{sec:solution} details the solution methodology. Section~\ref{sec:results} provides simulation results, followed by discussion in Section~\ref{sec:discussion} and conclusions in Section~\ref{sec:conclusion}.

\begin{table}[t]
\small
\centering
\caption{Summary of key symbols used in the paper.}
\label{tab:symbols}
\begin{tabularx}{\columnwidth}{@{}l>{\raggedright\arraybackslash}X@{}}
\toprule
\textbf{Symbol} & \textbf{Description} \\
\midrule
\(n\) & Number of channel uses \\
\(\eta\) & Channel transmittance (random variable) \\
\(\mu_{\ln(\eta)}, \sigma_{\ln(\eta)}\) & Mean and standard deviation of \(\ln(\eta)\) \\
\(\mu_\eta, \sigma_\eta\) & Mean and standard deviation of \(\eta\) \\
\(\overline{n}_B\) & Mean thermal photon number (random variable) \\
\(\mu_{n_B}, \sigma_{n_B}\) & Mean and standard deviation parameters of the \(\overline{n}_B\) distribution \\
\(\delta\) & Covertness threshold \\
\(c_{\text{cov}}\) & Covertness constant of the channel \\
\(D(\cdot\|\cdot)\) & Quantum relative entropy (QRE) \\
\(q, q^*\) & Transmission probability; its optimized value under stochastic uncertainty \\
\(R_{\text{ach}}, R^*\) & Instantaneous achievable rate under a realized channel; optimized chosen code rate \\
\(T, T^*\) & Per-use covert throughput \(T = qR\); optimal throughput \(T^* = q^*R^*\) \\
\bottomrule
\end{tabularx}
\end{table}

For clarity, we summarize the key notation used throughout the paper. 
The physical channel is characterized by its transmittance $\eta$ and mean thermal photon number $\overline{n}_B$, both modeled as random variables. 
For turbulence–induced loss, we use either log-domain parameters $(\mu_{\ln(\eta)}, \sigma_{\ln(\eta)})$ for a lognormal model, or linear-domain parameters $(\mu_\eta, \sigma_\eta)$ when modeling $\eta$ directly. 
For the background noise, we use $(\mu_{n_B}, \sigma_{n_B})$. 
The fundamental covert-communication limit is captured by the covertness constant \(c_{\text{cov}}\). 
Alice's transmission strategy is specified by the transmission probability \(q\) and the chosen code rate \(R\), with optimizer outputs \(q^*\) and \(R^*\). The instantaneous achievable rate supported by a realized channel is denoted by \(R_{\text{ach}}(\eta,\overline{n}_B)\). Our primary performance metric is the per-use covert throughput \(T = qR\); the optimizer returns \(T^* = q^*R^*\). 
Table~\ref{tab:symbols} summarizes the notation.

\section{Related Work}
\label{sec:related}

Covert communication has been extensively studied in non-quantum communication systems. The foundational works~\cite{bash2013limits, wang2016fundamental} established the \emph{SRL}, which states that over $n$ channel uses, the number of reliably and covertly transmissible bits scales at most as $\mathcal{O}(\sqrt{n})$~\cite{bash2013limits}. This limitation has motivated a large body of research on coding, modulation, and signaling strategies that maximize throughput under strict undetectability constraints, with applications spanning military links, privacy-preserving networking in civilian infrastructure, and in-body covert sensing/communication~\cite{padmal2025fat}.

More recently, attention has shifted to quantum channels, where the medium conveys quantum states of light rather than classical symbols. CQC leverages quantum effects (e.g., photon-counting statistics and bosonic noise models) to guarantee undetectability against an optimal quantum-limited adversary~\cite{anderson2024square, anderson2024covert, anderson2025achievability}. These works derive performance bounds for pure-loss, thermal-noise, and entanglement-assisted channels. While the $\mathcal{O}(\sqrt{n})$ scaling persists, the constants and operational regimes differ, opening new avenues for stronger security guarantees. However, most CQC analyses assume perfectly known and time-invariant channel parameters (typically the transmittance and mean thermal photon number). This assumption is rarely satisfied in practice (especially in free-space optical (FSO) and satellite links) where atmospheric turbulence, pointing error, and background illumination vary on short timescales.

In the classical literature on LPD, stochastic channel uncertainty has also been addressed by modeling fading/noise as random processes and optimizing average-case or outage-constrained performance~\cite{sobers2017covert, soltani2018covert, arghavani2023covert}. To the best of our knowledge, prior work in the \emph{quantum} setting has not developed a systematic CQC design framework under \emph{stochastic uncertainty} that explicitly couples distributed channel uncertainty with probabilistic constraints on both covertness failure and decoding failure. This gap is security-critical: in covert communication, uncertainty does not merely reduce performance, but can invalidate the intended protection of the communication event itself. Related classical robust-design work under bounded uncertainty has also shown that reliability and covertness need not be governed by the same adverse parameter realization, reinforcing the need to treat the two constraints separately in robust covert design~\cite{arghavani2026conflict}. Building on quantum covertness limits and our earlier bounded-uncertainty CQC model~\cite{arghavani2026robust}, we therefore propose a risk-aware CQC design methodology for stochastic channels, thereby bridging pessimistic worst-case designs, idealized perfect-knowledge models, and deployment-relevant security-risk calibration.

\section{System Model and Performance Metrics}
\label{sec:model}

\begin{figure*}[!t]
    \centering
    \includegraphics[width=0.7\linewidth]{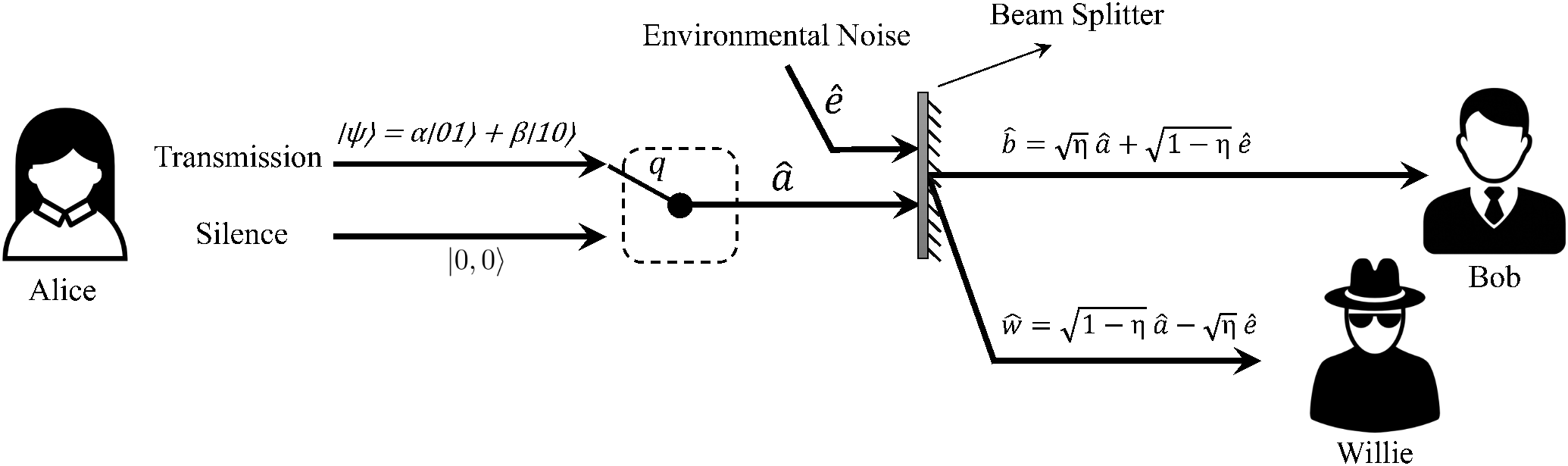}
    \caption{The CQC channel model, where Alice transmits a signal state ($\ket{\psi}$) with probability $q$ or vacuum with probability $1-q$. A beamsplitter with transmittance $\eta$ mixes her input $\hat{a}$ with environmental noise $\hat{e}$, directing outputs to Bob ($\hat{b}$) and an adversary, Willie ($\hat{w}$).
    }
    \label{fig:channel}
\end{figure*}

To analyze the impact of channel uncertainty, we first outline the system model and associated performance measures, originally introduced in~\cite{anderson2024covert, anderson2024square}. This model serves as the baseline for our stochastic risk-aware formulation and is reproduced here to make the paper self-contained.

\subsection{Physical Setup and Channel Model}
We consider a CQC scenario with three parties: a sender (Alice), a legitimate receiver (Bob), and a passive adversary (Willie). Alice’s goal is to deliver quantum information to Bob while preventing Willie from reliably detecting that communication is occurring. Time is organized into frames; in each frame, the optical channel is used \(n\) times. We refer to each channel use as a \emph{slot}. In slot \(t\in\{1,\ldots,n\}\), Alice transmits a signal state with probability \(q\), and remains silent (sending the vacuum state \(\ket{00}\)) with probability \(1-q\). The signal is a single qubit encoded using the \emph{dual-rail scheme}, a standard technique in linear optical quantum computation~\cite{anderson2024covert}. In this encoding, a logical qubit is represented by a single photon distributed across two optical modes. The logical basis states are
$\ket{0}_L = \ket{01}, \quad \ket{1}_L = \ket{10}$,
so a general logical qubit can be written as $\ket{\psi} = \alpha\ket{0}_L + \beta\ket{1}_L = \alpha\ket{01} + \beta\ket{10}$,
with normalization $|\alpha|^2 + |\beta|^2 = 1$.

The communication channel is modeled as a \emph{lossy bosonic channel with thermal noise}, appropriate for both fiber-optic and free-space links. As illustrated in Figure~\ref{fig:channel}, the channel is implemented by a beamsplitter with two input and two output modes. Alice’s input mode $\hat{a}$ is mixed with an environmental thermal mode $\hat{e}$ of mean photon number $\overline{n}_B$. One output mode, $\hat{b}$, is directed to Bob; the other, $\hat{w}$, is accessible to Willie. The corresponding annihilation operators satisfy
\begin{align}
    \hat{b} &= \sqrt{\eta}\,\hat{a} + \sqrt{1-\eta}\,\hat{e}, \\
    \hat{w} &= \sqrt{1-\eta}\,\hat{a} - \sqrt{\eta}\,\hat{e},
\end{align}
where $\eta \in (0,1)$ is the transmittance, i.e., the fraction of Alice’s signal that reaches Bob. The complementary fraction $1-\eta$ is diverted to Willie. The environmental input $\hat{e}$ is assumed to be a zero-mean thermal state $\hat{\rho}_{\overline{n}_B}$ with mean photon number $\overline{n}_B$. In the Fock basis,
\[
\hat{\rho}_{\overline{n}_B} = \sum_{k=0}^\infty \frac{(\overline{n}_B)^k}{(1+\overline{n}_B)^{k+1}} \ket{k}\bra{k}.
\]
This model captures the fundamental interaction between signal loss and environmental noise that governs Bob’s decoding performance and Willie’s detection capability.

\subsection{Covertness Requirement and Security Analysis}

From Willie's standpoint, detection is a binary quantum hypothesis test between:
\begin{itemize}
    \item \textbf{Hypothesis $H_0$:} Alice remains silent throughout the entire frame. Willie observes the frame-level state
    \[
    \hat{\rho}_{0,\text{frame}}^W = \big(\hat{\rho}_{0,\text{slot}}^W\big)^{\otimes n},
    \]
    where $\hat{\rho}_{0,\text{slot}}^W$ is the noise state per slot (i.e., per channel use).

    \item \textbf{Hypothesis $H_1$:} Alice transmits according to her probabilistic scheme. Willie observes the corresponding frame-level state $\hat{\rho}_{1,\text{frame}}^W$.
\end{itemize}

Willie’s optimal detector incurs \emph{false alarms} (type I error) with probability $P_{\text{FA}}$ and \emph{missed detections} (type II error) with probability $P_{\text{MD}}$. Under equal priors $P(H_0)=P(H_1)=1/2$, the average error probability is
\[
P_e = \tfrac{1}{2}\big(P_{\text{FA}} + P_{\text{MD}}\big).
\]
The covertness criterion~\cite{arghavani2023covert} requires that Willie's minimum error probability be close to random guessing:
\begin{equation}\label{eq:covertness_1}
P_{e,\min} \ge \tfrac{1}{2} - \delta,
\end{equation}
for some small $\delta > 0$. The Helstrom bound~\cite{helstrom1976quantum} relates $P_{e,\min}$ to the trace distance between the two possible global states:
\[
P_{e,\min} = \tfrac{1}{2}\left(1 - \tfrac{1}{2}\,\big\| \hat{\rho}_{1,\text{frame}}^W - \hat{\rho}_{0,\text{frame}}^W \big\|_1 \right).
\]
Combining with \eqref{eq:covertness_1} yields the trace-distance bound
\begin{equation}
\label{eq:trace_dist_bound}
\big\| \hat{\rho}_{1,\text{frame}}^W - \hat{\rho}_{0,\text{frame}}^W \big\|_1 \le 4\delta.
\end{equation}

By the quantum Pinsker inequality~\cite{anderson2024square, wilde2017quantum},
\[
\tfrac{1}{2}\big\|\rho - \sigma\big\|_1 \;\le\; \sqrt{\tfrac{1}{2}\,D(\rho\|\sigma)},
\]
so \eqref{eq:trace_dist_bound} is satisfied if the quantum relative entropy (QRE) obeys
\[
D\!\left(\hat{\rho}_{1,\text{frame}}^W \,\big\|\, \hat{\rho}_{0,\text{frame}}^W\right) \le 8\delta^2.
\]
Under the transmission model adopted here, and \emph{conditional on a realized frame-level channel state} $(\eta,\overline{n}_B)$, Willie’s observations across slots are modeled as identically distributed and independent under each hypothesis. Accordingly, for a fixed frame state, the frame-level states are
\[
\hat{\rho}_{1,\text{frame}}^W(\eta,\overline{n}_B)
=
\big(\hat{\rho}_{1,\text{slot}}^W(\eta,\overline{n}_B)\big)^{\otimes n},
\qquad
\hat{\rho}_{0,\text{frame}}^W(\eta,\overline{n}_B)
=
\big(\hat{\rho}_{0,\text{slot}}^W(\eta,\overline{n}_B)\big)^{\otimes n},
\]
and the QRE reduces to
\[
D\!\left(
\hat{\rho}_{1,\text{frame}}^W(\eta,\overline{n}_B)
\,\big\|\,
\hat{\rho}_{0,\text{frame}}^W(\eta,\overline{n}_B)
\right)
=
n\,D\!\left(
\hat{\rho}_{1,\text{slot}}^W(\eta,\overline{n}_B)
\,\big\|\,
\hat{\rho}_{0,\text{slot}}^W(\eta,\overline{n}_B)
\right).
\]
This is the standard product-state reduction for a quasi-static frame with conditional i.i.d.\ slots, rather than a separate claim about the channel being entanglement-breaking. Throughout the stochastic outage analysis developed later, we interpret this conditional-on-state formulation under the conservative threat model that Willie’s detection capability for a given frame is evaluated with respect to the realized frame state \((\eta,\overline{n}_B)\), and the resulting outage probability is then computed over the distribution of such frame states across frames/sessions. In Section~\ref{sec:channel-models}, we then place a probability law on the frame-level state $(\eta,\overline{n}_B)$ and study outage probabilities across frames/sessions. When the transmission probability $q$ is small, the per-slot state under $H_1$, $\hat{\rho}_{1,\text{slot}}^W$, is close to $\hat{\rho}_{0,\text{slot}}^W$. In this sparse-transmission regime, the per-slot QRE can be approximated using the quantum $\chi^2$-divergence, yielding
\[
D\!\left(\hat{\rho}_{1,\text{slot}}^W \,\big\|\, \hat{\rho}_{0,\text{slot}}^W\right) \lesssim
q^2 \, \frac{(1-\eta)^2}{\eta \overline{n}_B \,\big(1+\eta \overline{n}_B\big)}.
\]
Substituting into the QRE bound gives the covertness condition
\[
n \, q^2 \, \frac{(1-\eta)^2}{\eta \overline{n}_B \,\big(1+\eta \overline{n}_B\big)} \;\lesssim\; 8\delta^2.
\]
Defining the \emph{covertness constant}
\begin{equation}
\label{eq:ccov}
c_{\text{cov}} \;=\; \frac{\sqrt{2\,\eta \,\overline{n}_B \,\big(1+\eta \overline{n}_B\big)}}{1-\eta},
\end{equation}
the constraint on the transmission probability becomes
\[
q \;\le\; \frac{2\delta\,c_{\text{cov}}}{\sqrt{n}}.
\]
A larger $c_{\text{cov}}$ permits a higher feasible $q$ (and thus greater throughput) without compromising covertness.

\subsection{Reliability Requirement}

When Alice transmits, Bob must be able to recover the encoded qubit from his noisy measurements reliably. Bob's decoding performance is affected by two primary factors:
\begin{enumerate}
    \item \emph{Channel-induced errors}: photon loss, governed by $\eta$, and thermal noise, characterized by $\overline{n}_B$;
    \item \emph{Projection failure}: decoding via projection onto the dual-rail computational subspace $\{\ket{01}, \ket{10}\}$ can discard the photon even if it arrives.
\end{enumerate}

Following~\cite{anderson2024square}, the cumulative impact of these effects is well modeled as a depolarizing channel acting on the input qubit:
\[
\mathcal{E}(\rho_{\text{in}}) \;=\; (1-p)\,\rho_{\text{in}} \;+\; p\,\frac{I}{2},
\]
where $\rho_{\text{in}}$ is the qubit sent by Alice, and the depolarizing probability $p$ is
\begin{equation}
\label{eq:p}
p \;=\; 1 \;-\; \frac{\eta}{\big[\,1+(1-\eta)\overline{n}_B\,\big]^4}.
\end{equation}
The corresponding Pauli error probabilities are
\[
\vec{p} \;=\; \big[ p_I, p_X, p_Y, p_Z \big] \;=\; \Big[ 1-\tfrac{3p}{4}, \tfrac{p}{4}, \tfrac{p}{4}, \tfrac{p}{4} \Big].
\]
The achievable communication rate per transmitted qubit, from the hashing bound~\cite{anderson2024square}, is
\begin{equation}
\label{eq:comm_rate}
R_{\text{ach}}(\eta,\overline n_B) \;=\; \big(1-H(\vec p)\big)^+,
\end{equation}
where $H(\vec{p})=-\sum_i p_i\log_2 p_i$ is the Shannon entropy of the error distribution (with $i \in \{I,X,Y,Z\}$) and $(x)^+=\max(x,0)$. Expanding $H(\vec{p})$ gives
\[
H(\vec{p}) \;=\; -\Big(1 - \tfrac{3p}{4}\Big) \log_2\!\Big(1 - \tfrac{3p}{4}\Big)
\;-\; 3 \cdot \tfrac{p}{4}\,\log_2\!\Big(\tfrac{p}{4}\Big).
\]

Combining the instantaneous achievable rate in \eqref{eq:comm_rate} with the covertness-induced constraint on $q$ characterizes the secure operating region of the system. For convenience, we also define the \emph{per-use covert throughput}
\begin{equation}
\label{eq:throughput}
T \;\triangleq\; q\,R,
\end{equation}
which we will optimize under stochastic channel uncertainty in subsequent sections.

\section{Channel 
Models and Assumptions}
\label{sec:channel-models}

We now formalize the statistical models, knowledge assumptions, and risk constraints that underpin our analysis. Our framework moves from fixed channel parameters to a physically faithful \emph{stochastic} model that enables risk-aware performance optimization.

\subsection{Stochastic Channel Parameters}

In Section~\ref{sec:model}, we treated $\eta$ and $\overline{n}_B$ as known constants. In realistic deployments, these quantities vary due to environmental dynamics. We therefore model them as random variables with known probability distributions, obtained via prior measurement or physical modeling.

\begin{itemize}
    \item \textbf{Transmittance $\eta$:} For FSO and satellite links, the dominant fluctuation mechanism is atmospheric turbulence, which induces scintillation (fading). Because the optical transmittance in our bosonic channel model must satisfy $0<\eta\le 1$, we model $\eta$ using a \emph{truncated lognormal} distribution on $(0,1]$:
    \[
    \eta \sim \mathrm{Lognormal}_{(0,1]}\!\left(\mu_{\ln(\eta)},\,\sigma^2_{\ln(\eta)}\right),
    \]
    with density
    \[
    f_\eta(x)
    =
    \frac{f_{\mathrm{LN}}(x;\mu_{\ln(\eta)},\sigma_{\ln(\eta)})}
    {F_{\mathrm{LN}}(1;\mu_{\ln(\eta)},\sigma_{\ln(\eta)})}
    \,\mathbf{1}_{(0,1]}(x),
    \]
    where $f_{\mathrm{LN}}$ and $F_{\mathrm{LN}}$ are the PDF and CDF of the corresponding untruncated lognormal law. This preserves the usual lognormal motivation from wave-propagation theory while enforcing the physical support required by the channel model and the covertness expression~\cite{andrews2023laser, khalighi2014survey}.
    
    \item \textbf{Thermal noise $\overline{n}_B$:} Background photons arise from sources such as celestial radiation, local equipment, and detector dark counts. We model $\overline{n}_B$ as a \emph{truncated Gaussian}~\cite{alexander1997optical}:
    \[
    \overline{n}_B \sim \mathcal{N}_{\;\,[0,b]}\!\left(\mu_{n_B},\, \sigma^2_{n_B}\right),
    \]
    where truncation enforces the physical constraint $\overline{n}_B \ge 0$ and optionally caps extreme values at $b$ to model sensor saturation or environmental maxima.
\end{itemize}

For analytical tractability, we assume $\eta$ and $\overline{n}_B$ are statistically independent, a reasonable approximation when turbulence-induced fading and background illumination have distinct physical origins. Scenarios with coupling (e.g., dense fog) are left for future work.

Throughout this paper, we interpret the random pair \((\eta,\overline{n}_B)\) as a \emph{frame-level} channel state: for any given frame of \(n\) channel uses, the parameters \(\eta\) and \(\overline{n}_B\) are taken to be constant across the \(n\) slots, but are randomly redrawn from their governing distributions across different frames/sessions. As a result, the probabilistic constraints introduced later are imposed over the distribution of these quasi-static frame states, rather than over independently varying slot-level channel realizations. Equivalently, the product-form expressions in Section~\ref{sec:model} should be read as holding \emph{conditional on a realized frame state} \((\eta,\overline{n}_B)\), while the risk constraints later average over the distribution of these frame states across frames.

\subsection{Operational Knowledge and Strategy Constraints}

We specify the information available to each party:
\begin{itemize}
    \item \textbf{Alice and Bob:} Know the distributions $f_\eta$ and $f_{\overline{n}_B}$, but not the instantaneous realizations during transmission. This knowledge is acquired before the covert session via non-covert probing, long-term environmental sensing, or predictive modeling.
    \item \textbf{Willie:} We adopt a conservative threat model in which Willie knows the same channel laws as Alice and Bob and, for each frame, his detection capability is evaluated conditional on the realized frame-level state \((\eta,\overline{n}_B)\). Operationally, this corresponds to a warden with per-frame CSI or sufficiently accurate per-frame estimation of the realized state.
    \item \textbf{No instantaneous CSI:} During covert transmission, Alice receives no real-time channel state information (CSI). Feedback is impractical or would itself risk revealing activity.
\end{itemize}

This is the threat model used throughout the outage analysis in the paper. Accordingly, the covertness-outage event is defined by asking whether the chosen transmission probability \(q\) violates the adopted QRE-based covertness surrogate for the realized frame state, and the probability in the corresponding risk constraint is then taken over the distribution of frame states across frames/sessions. We do not analyze the alternative model in which Willie knows only the channel distribution but not the realized frame state and must instead perform a mixture-state hypothesis test. That leads to a different covertness formulation and is outside the scope of the present paper.

All probabilistic guarantees in this paper are \emph{conditional on the assumed channel laws} \(f_\eta\) and \(f_{\overline{n}_B}\). In other words, the risk budgets \((\epsilon_{\text{cov}},\epsilon_{\text{rel}})\), the induced quantiles, and the resulting operating points are calibrated with respect to the stipulated stochastic model. If the true channel statistics differ materially from the assumed distributions, the realized outage probabilities may no longer match the nominal budgets. Our results should therefore be interpreted as \emph{model-based risk guarantees}, not distribution-free guarantees.

\subsection{Fixed Transmission Strategy, Risk Budgets, and Outage Events}

In the absence of real-time CSI, Alice commits to a \emph{single fixed transmission strategy} specified by $(q, R)$ for an entire frame, where $q \in [0,1]$ is the per-channel use transmission probability and $R \in [0,1]$ is the chosen quantum code rate (design variable). Allowing $R=0$ explicitly includes the trivial zero-payload operating point, which is useful when the outage budgets are so stringent that no nonzero code rate is feasible. Let
\begin{equation}
\label{eq:Rach}
R_{\text{ach}}(\eta,\overline{n}_B) \;\triangleq\; \big(1 - H(\vec{p}(\eta,\overline{n}_B))\big)^+
\end{equation}
denote the \emph{instantaneous achievable rate} under the realized channel, given in~\eqref{eq:comm_rate}. We also recall the \emph{per-use covert throughput} $T \triangleq q\,R$, with optimal value $T^* = q^* R^*$. Here, $T$ is a \emph{design throughput}: it quantifies the scheduled payload per channel use under the fixed strategy $(q,R)$. The randomness of the channel enters separately through the covertness and reliability outage events, which are controlled by the risk budgets $\epsilon_{\text{cov}}$ and $\epsilon_{\text{rel}}$.

\paragraph{Risk budgets.}
We introduce two user-specified probability thresholds, called \emph{risk budgets}:
\[
\epsilon_{\text{cov}},\, \epsilon_{\text{rel}} \in (0,1).
\]
Here, $\epsilon_{\text{cov}}$ is the maximum tolerable probability (per frame, over the randomness of $(\eta,\overline{n}_B)$) that the covertness constraint is violated. $\epsilon_{\text{rel}}$ is the maximum tolerable probability (per frame) that decoding fails, i.e., that the chosen code rate $R$ exceeds the instantaneous achievable rate $R_{\text{ach}}(\eta,\overline{n}_B)$.

We consider:
\begin{itemize}
    \item \textbf{Covertness outage:} A security failure that occurs when the realized channel is sufficiently clear (equivalently, $c_{\text{cov}}(\eta,\overline{n}_B)$ is too small) so that the chosen $q$ violates the covertness constraint.
    \item \textbf{Reliability outage:} A communication failure that occurs when the realized channel is too noisy for Bob to decode at rate $R$, i.e., when $R_{\text{ach}}(\eta,\overline{n}_B) < R$.
\end{itemize}

Both outage probabilities must lie below their budgets:
\[
\mathbb{P}[\text{Covertness outage}] = \mathbb{P}\!\left[\, q > \frac{2\delta}{\sqrt{n}}\; c_{\text{cov}}(\eta,\overline{n}_B) \right] \;\le\; \epsilon_{\text{cov}},\]
\[\mathbb{P}[\text{Reliability outage}] = \mathbb{P}\!\left[\, R > R_{\text{ach}}(\eta,\overline{n}_B) \right] \;\le\; \epsilon_{\text{rel}}.
\]

The covertness-outage event above is induced by the QRE-based sufficient condition and sparse-transmission approximation introduced in Section~\ref{sec:model}. Accordingly, the budget \(\epsilon_{\text{cov}}\) should be interpreted as controlling outage relative to this adopted covertness surrogate under the quasi-static stochastic model of the paper. In the risk-constrained parameter regimes emphasized in our main simulations, the resulting optimizers remain in the sparse-transmission regime; outside that regime, the covertness constraint should be interpreted with corresponding caution.

To avoid ambiguity, we distinguish the following quantities. The \emph{instantaneous achievable rate} \(R_{\text{ach}}(\eta,\overline{n}_B)\) is the rate supported by a \emph{specific} realization of the channel parameters \((\eta,\overline{n}_B)\). As such, it is a random variable induced by the channel uncertainty. In simulations, we draw \(K\) realizations and record \(r_i \coloneqq R_{\text{ach}}(\eta_i,\overline{n}_{B,i})\) for \(i=1,\dots,K\). The \emph{chosen code rate} \(R\) is a single design variable that Alice fixes for the entire transmission (frame/session). A \emph{reliability outage} occurs whenever the chosen rate exceeds what the realized channel can support, i.e., when \(R > R_{\text{ach}}(\eta,\overline{n}_B)\).

\section{Problem Formulation}
\label{sec:problem}

The core design problem is to choose a fixed transmission strategy $(q,R)$ that balances design throughput against the risks of security and reliability failures. Our primary formulation is a \emph{risk-constrained} model that maximizes performance subject to hard probability budgets. We then study a secondary \emph{risk-adjusted} extension that internalizes outage costs directly in the objective and is used only for exploratory policy analysis rather than for guarantee-bearing design.

\subsection{Risk-Constrained Formulation}

We first seek the strategy that maximizes per-use covert throughput while ensuring that the outage probabilities induced by the adopted stochastic covertness-and-reliability model stay below prescribed budgets. This formulation fits operational settings with non-negotiable limits on failure (e.g., ``the probability of covertness outage must be below $0.1\%$'' or ``the probability of reliability outage must be below $1\%$''). Given the definition of $R_{\text{ach}}$ presented in Eq.\eqref{eq:Rach}, the optimization problem is formally stated as:
\begin{align}
\max_{q,\,R}\quad & T(q, R) = q\,R \label{eq:objective_constrained}\\
\text{s.t.}\quad &
\mathbb{P}\!\left[q > \frac{2\delta}{\sqrt{n}}\; c_{\text{cov}}(\eta,\overline{n}_B)\right] \;\le\; \epsilon_{\text{cov}}, \label{eq:cov_outage}\\
& \mathbb{P}\!\left[R > R_{\text{ach}}(\eta,\overline{n}_B)\right] \;\le\; \epsilon_{\text{rel}}, \label{eq:rel_outage}\\
& 0 \le q \le 1,\;\; 0 \le R \le 1. \nonumber
\end{align}
Here, the objective \eqref{eq:objective_constrained} is the \emph{per-use covert throughput} $T=qR$; the probabilities in \eqref{eq:cov_outage} and \eqref{eq:rel_outage} are taken over the randomness of $(\eta,\overline{n}_B)$. The chosen code rate, $R$, is a decision variable in the optimization, and we denote its optimal value as $R^*$. As we will show in Theorem~\ref{thm:optimal_throughput}, in the continuous regime emphasized in this paper the optimizer \((q^*,R^*)\) is characterized by quantiles of the induced distributions of \(c_{\text{cov}}(\eta,\overline{n}_B)\) and \(R_{\text{ach}}(\eta,\overline{n}_B)\), with the achievable-rate quantile denoted by \(R_{\max}\). To the best of our knowledge, this is among the first CQC formulations to treat physical-layer uncertainty stochastically while jointly optimizing throughput under explicit probabilistic risk constraints.

\subsection{Exploratory Extension: Risk-Adjusted Throughput}
\label{subsec:risk_adjusted}

While the risk-constrained model provides hard guarantees, it is often useful to encode risk preferences directly in the objective, as is common in financial engineering. We therefore define the \emph{risk-adjusted throughput}
\begin{align}
J(q,R)
&\;\triangleq\; T(q,R)
-\lambda_{\text{cov}}\;\mathbb{P}\!\left[q > \tfrac{2\delta}{\sqrt{n}}\; c_{\text{cov}}(\eta,\overline{n}_B)\right]\nonumber\\
&-\lambda_{\text{rel}}\;\mathbb{P}\!\left[R > R_{\text{ach}}(\eta,\overline{n}_B)\right],
\label{eq:J(q,R)}
\end{align}
where $T(q,R)=qR$ is the per-use throughput and $\lambda_{\text{cov}},\lambda_{\text{rel}}\!\ge 0$ are \emph{risk-aversion parameters} that penalize covertness and reliability outages, respectively. Larger $\lambda$ values reflect greater aversion to the corresponding failure mode. We then solve the following problem:
\begin{align}
\max_{q,\,R}\quad & J(q,R)\\
\text{s.t.}\quad & 0 \le q \le 1,\;\; 0 \le R \le 1. \nonumber
\end{align}

This formulation selects a single strategy that balances reward (throughput) against weighted risk; by sweeping $(\lambda_{\text{cov}},\lambda_{\text{rel}})$, one explores the trade-off induced by the weighted objective. Unlike the risk-constrained formulation, however, it does \emph{not} impose hard probabilistic guarantees on covertness or reliability. It should therefore be interpreted as a decision-theoretic design tool for exploring operating-point preferences, rather than as the primary formulation for guarantee-critical deployments. Theorem~\ref{thm:foc_risk_adjusted} characterizes the first-order conditions satisfied by any differentiable interior optimum of this formulation.

\subsection{Interpretation: The Risk–Performance Frontier}

Our stochastic framework generalizes prior perfect-knowledge models by characterizing the trade-off between performance and risk through a \emph{Pareto-frontier} viewpoint. A strategy is Pareto optimal if no other strategy improves one objective (e.g., throughput) without worsening at least one risk measure.

\begin{itemize}
   \item \textbf{Risk-constrained model:} Lets a designer \emph{select} a point on the frontier by specifying risk budgets $(\epsilon_{\text{cov}},\epsilon_{\text{rel}})$. \textbf{Theorem~\ref{thm:optimal_throughput}} characterizes the optimizer $(q^*,R^*)$ under these constraints and its dependence on the budgets. Varying $(\epsilon_{\text{cov}},\epsilon_{\text{rel}})$ traces the constrained frontier.
\item \textbf{Risk-adjusted model:} Explores points on the frontier by \emph{pricing} outages via $(\lambda_{\text{cov}},\lambda_{\text{rel}})$. \textbf{Theorem~\ref{thm:foc_risk_adjusted}} gives only the first-order conditions for differentiable interior optima. In practice, we explore this formulation numerically through a weighted-sum scalarization. For a nonconvex frontier, this need not recover every Pareto-optimal point.
\end{itemize}

In both formulations, there is no single universally best strategy, only operating-point trade-offs, with the risk-constrained model serving as the primary guarantee-bearing design tool. Our approach characterizes this trade-off frontier under the adopted model, enabling mission-specific choices instead of defaulting to a single conservative worst-case design.

\section{Solution Methodology}
\label{sec:solution}

The key methodological step is to reduce the stochastic CQC design problem from physical-layer uncertainty in \((\eta,\overline{n}_B)\) to calibrated covertness- and reliability-outage variables induced by \(c_{\text{cov}}(\eta,\overline{n}_B)\) and \(R_{\text{ach}}(\eta,\overline{n}_B)\). Once this reduction is made, the risk-constrained formulation admits a closed-form optimizer through quantile bounds, while the risk-adjusted extension is explored numerically. The risk-constrained objective $T(q,R)=qR$ is bilinear\footnote{%
$T(q,R)$ has $\nabla^2 T=\begin{psmallmatrix}0&1\\[2pt]1&0\end{psmallmatrix}$ (eigenvalues $\pm1$), hence it is neither convex nor concave on any open convex set. 
Likewise, the risk-adjusted objective $J(q,R)$ has Hessian 
$\nabla^2 J=\begin{psmallmatrix}
-\lambda_{\text{cov}}\tfrac{n}{4\delta^2}f'_{c_{\text{cov}}}(u) & 1\\[2pt]
1 & -\lambda_{\text{rel}}f'_{R_{\text{ach}}}(R)
\end{psmallmatrix}$ (with $u=\tfrac{q\sqrt{n}}{2\delta}$). 
At points where $f'_{c_{\text{cov}}}(u)=f'_{R_{\text{ach}}}(R)=0$ (e.g., interior modes of common unimodal PDFs), 
$\nabla^2 J=\begin{psmallmatrix}0&1\\[2pt]1&0\end{psmallmatrix}$ is indefinite; 
hence $J$ is also nonconvex.} (and therefore neither convex nor concave), while the risk-adjusted objective $J(q,R)$ is generally nonconvex for the reasons summarized in the footnote. Moreover, the constraints are probabilistic in $(\eta,\overline{n}_B)$. Thus, generic gradient-based solvers do not directly exploit the structure used here. We therefore develop practical methods that yield the optimal transmission strategy $(q^*,R^*)$ for the risk-constrained formulation together with a principled numerical strategy for the risk-adjusted formulation.

\subsection{Solving the Risk-Constrained Problem}

Our approach is to convert the probabilistic constraints in \eqref{eq:cov_outage}--\eqref{eq:rel_outage} into deterministic one-sided quantile bounds. 
Because the outage events are defined using strict inequalities, the exact quantile object for general distributions is not the usual inverse CDF, but the largest threshold whose strict lower-tail probability remains within the prescribed risk budget.

For a real-valued random variable $X$ and $\epsilon\in(0,1)$, define the
strict-outage quantile
\begin{equation}
\label{eq:strict_outage_quantile}
Q_X^{<}(\epsilon)
\triangleq
\sup\left\{x\in\mathbb{R}: \mathbb{P}[X<x]\le \epsilon\right\}.
\end{equation}
Equivalently, if $F_X$ denotes the CDF of $X$, then
\begin{equation}
\label{eq:strict_outage_quantile_cdf}
Q_X^{<}(\epsilon)
=
\sup\left\{x\in\mathbb{R}: F_X(x^-)\le \epsilon\right\},
\end{equation}
where
\begin{equation}
\label{eq:cdf_left_limit}
F_X(x^-)
\triangleq
\lim_{y\uparrow x} F_X(y).
\end{equation}
The definition in \eqref{eq:strict_outage_quantile} correctly handles atoms,
flat CDF regions, and boundary cases induced by the strict outage events.

\begin{theorem}[Optimal risk-constrained throughput]
\label{thm:optimal_throughput}
For the risk-constrained program \eqref{eq:objective_constrained}–\eqref{eq:rel_outage}, define the strict-outage distribution function
\[
F_X^{<}(x)\;\triangleq\;\mathbb{P}[X<x],
\]
and the corresponding strict-outage quantile
\[
Q_X^{<}(\epsilon)\;\triangleq\;\sup\{x:\,F_X^{<}(x)\le \epsilon\}.
\]
Then the maximum per-use covert throughput is
\[
T^* \;=\; q^* R^* \;=\; q_{\max}\,R_{\max},
\]
with
\[
q_{\max} \;=\; \min\!\left\{1,\; \frac{2\delta}{\sqrt{n}}\;Q_{c_{\text{cov}}}^{<}(\epsilon_{\text{cov}})\right\},
\qquad
R_{\max} \;=\; Q_{R_{\text{ach}}}^{<}(\epsilon_{\text{rel}}).
\]
This characterization is exact for arbitrary distributions, including cases with atoms at the relevant outage thresholds.
\end{theorem}

\begin{corollary}[Continuous-case reduction]
\label{cor:continuous_quantile_reduction}
If the relevant outage thresholds lie at continuity points of the CDFs of \(c_{\text{cov}}\) and \(R_{\text{ach}}\), then Theorem~\ref{thm:optimal_throughput} reduces to
\[
T^* \;=\; q^* R^* \;=\; q_{\max}\,R_{\max},
\]
with
\begin{equation}
\label{eq:q_max}
q_{\max} \;=\; \min\!\left\{1,\; \frac{2\delta}{\sqrt{n}}\;F^{-1}_{\,c_{\text{cov}}}(\epsilon_{\text{cov}})\right\},
\end{equation}
\begin{equation}
\label{eq:R_max}
R_{\max} \;=\; F^{-1}_{\,R_{\text{ach}}}(\epsilon_{\text{rel}}),
\end{equation}
where \(F_X^{-1}(\cdot)\) denotes the usual left-continuous quantile function.
\end{corollary}

\noindent\textit{Proof.} Deferred to~\ref{Proof:thm_risk_constrained}. The proof works directly with the strict-outage events and therefore does not require continuity assumptions. Corollary~\ref{cor:continuous_quantile_reduction} then recovers the simpler \(F^{-1}\)-based expression used in the continuous regime emphasized in our simulations.
\medskip

In practice, for the continuous regime emphasized in our simulations, the required quantiles are computed numerically from Monte Carlo samples of \((\eta,\overline{n}_B)\) as ordinary empirical quantiles (see Section~\ref{sec:results}).

\subsection{Characterizing the Risk-Adjusted Formulation via First-Order Conditions}
For the risk-adjusted objective \eqref{eq:J(q,R)}, we do not derive a closed-form global optimizer. Instead, we characterize differentiable interior optima via first-order conditions. Let
\[
c_{\text{cov}} \;=\; c_{\text{cov}}(\eta,\overline{n}_B), 
\qquad
R_{\text{ach}} \;=\; R_{\text{ach}}(\eta,\overline{n}_B),
\]
and denote by $F_{c_{\text{cov}}}$, $F_{R_{\text{ach}}}$ their CDFs and by $f_{c_{\text{cov}}}$, $f_{R_{\text{ach}}}$ their PDFs (when they exist; otherwise interpret conditions in the subgradient sense).

\begin{theorem}[First-order conditions for differentiable interior optima]
\label{thm:foc_risk_adjusted}
Fix $\lambda_{\text{cov}},\lambda_{\text{rel}}\!\ge 0$. Any interior maximizer $(q^*,R^*)$ of the risk-adjusted objective \eqref{eq:J(q,R)} that is differentiable at $(q^*,R^*)$ must satisfy
\begin{align}
R^* \;&=\; \lambda_{\text{cov}}\;\frac{\sqrt{n}}{2\delta}\;
f_{c_{\text{cov}}}\!\left(\frac{q^* \sqrt{n}}{2\delta}\right), \label{eq:foc_q}\\
q^* \;&=\; \lambda_{\text{rel}}\; f_{R_{\text{ach}}}(R^*). \label{eq:foc_R}
\end{align}
\end{theorem}

\noindent\textit{Proof.} Deferred to~\ref{Proof:thm_foc_risk_adjusted}. The result follows by writing the objective as
\[
J(q,R)=qR-\lambda_{\text{cov}}F_{c_{\text{cov}}}\!\left(\tfrac{q\sqrt{n}}{2\delta}\right)
-\lambda_{\text{rel}}F_{R_{\text{ach}}}(R),
\]
and setting the partial derivatives to zero, using $\frac{d}{dx}F_X(x)=f_X(x)$ and the chain rule.

\begin{figure}[b]
    \centering
    \includegraphics[width=\linewidth]{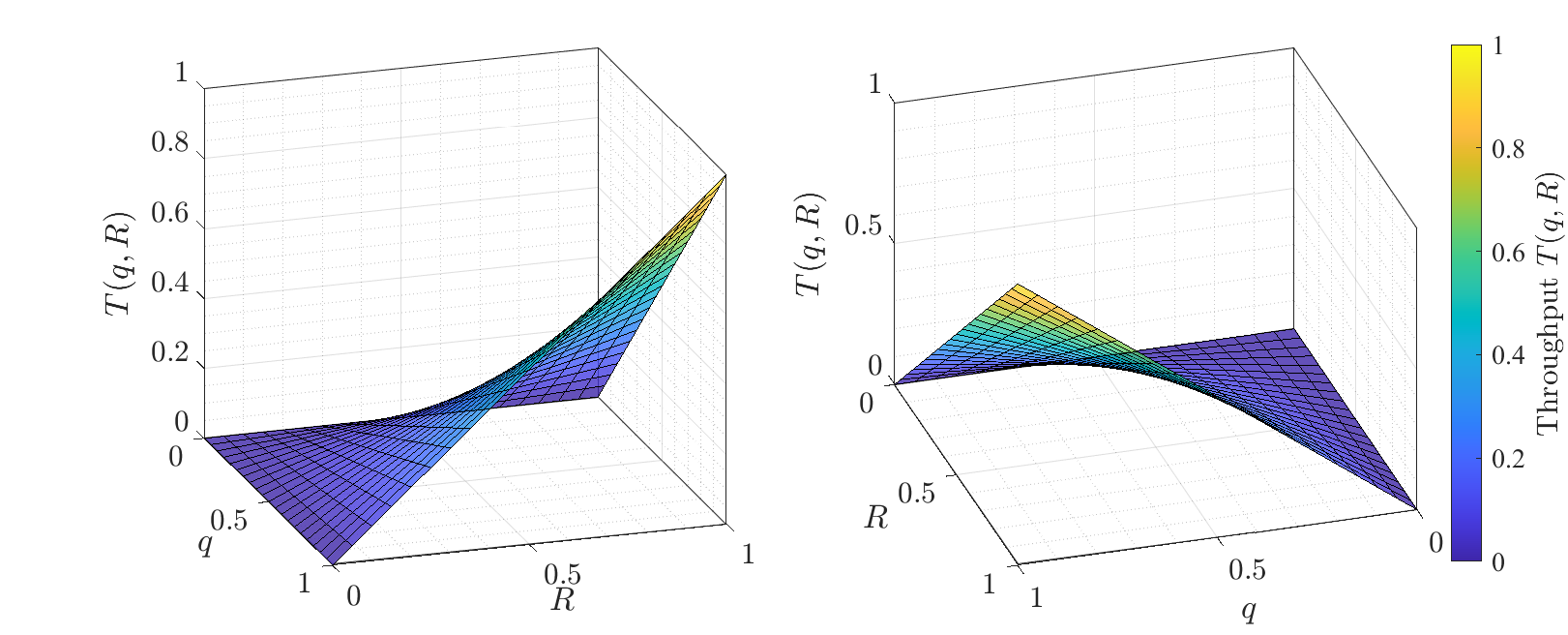}
    \caption{The per-use covert throughput $T(q,R)=qR$ is bilinear and therefore neither convex nor concave, which motivates the tailored solution methods used in this paper.}
    \label{fig:saddle}
\end{figure}

\subsection{Justification and Practical Considerations}

Figure~\ref{fig:saddle} illustrates the bilinear geometry of the throughput objective $qR$. However, the main simplification in the risk-constrained formulation does not come merely from observing nonconvexity; it comes from the CQC-specific reduction of stochastic physical-layer uncertainty to calibrated outage variables, which converts the chance constraints into quantile bounds and the feasible set into a rectangle in \((q,R)\). This yields the optimizer directly. For the risk-adjusted formulation, no analogous closed-form reduction is available, so we evaluate the weighted objective numerically on a dense grid.

\paragraph{Robustness and implementation.}
\begin{itemize}
    \item \emph{Risk-constrained:} In the continuous regime emphasized in this paper, the quantile-based bounds yield the global maximizer $(q_{\max},R_{\max})$ without iterative search. In practice, the required quantiles are estimated via Monte Carlo from the stipulated distributions of $(\eta,\overline{n}_B)$.
    \item \emph{Risk-adjusted:} Eqs.~\eqref{eq:foc_q} and \eqref{eq:foc_R} provide structural conditions for interior optima and help interpret how the optimizer depends on the two penalties. In the numerical results, however, we evaluate the risk-adjusted objective directly on a dense $(q,R)$ grid, since Monte Carlo-estimated outage probabilities induce step-like empirical CDFs and can lead to boundary-pinned solutions. Because the covertness surrogate used in this paper is derived from a sparse-transmission (\emph{small}-$q$) approximation, operating points with large $q$ should be interpreted qualitatively rather than as physically meaningful covert operating points.
\end{itemize}
Both methods are not tied to a specific parametric distribution and are parallelizable. They apply directly under the independence assumption in Section~\ref{sec:channel-models}.

\section{Analytical and Simulation Results}
\label{sec:results}

To evaluate our optimization framework and quantify the trade-off between covert throughput and risk, we run large-scale Monte Carlo experiments. We choose physically motivated parameter settings intended to represent FSO-inspired operating scenarios, and report performance landscapes that illustrate how a designer can select operating points under stochastic uncertainty.

\subsection{Simulation Parameters}
To highlight the fundamental behavior of our risk-aware framework, we begin with a high-quality baseline stochastic channel model and later stress the system with more challenging conditions. Unless stated otherwise, we consider a CQC system over an FSO link with the following settings:
\begin{itemize}
    \item \textbf{Channel uses ($n$):} $n = 10^7$ in most scenarios. This places the operation in a large-blocklength regime in which the $\sqrt{n}$ covertness scaling is the relevant asymptotic benchmark, while finite-size corrections are not modeled explicitly in the present study.
    \item \textbf{Security threshold ($\delta$):} $\delta = 0.05$, corresponding to the standard covertness target \(P_{e,\min} \ge \tfrac{1}{2}-\delta\) used throughout the paper.
    \item \textbf{Transmittance Model ($\eta$):} We model transmittance using a truncated lognormal distribution on $(0,1]$. In the baseline scenario, the underlying pre-truncation lognormal parameters are $\mu_{\ln(\eta)}=-0.0126$ and $\sigma_{\ln(\eta)}=0.05$, and samples outside $(0,1]$ are rejected (equivalently, resampled from the truncated law). This yields a physically admissible high-quality channel with low volatility. Because truncation changes the moments, all reported numerical results are computed from the truncated samples actually used in the simulations.
    \item \textbf{Thermal-noise model ($\overline{n}_B$):} $\overline{n}_B \sim \mathcal{N}_{[0,\,0.5]}\!\big(\mu_{n_B},\,\sigma_{n_B}^2\big)$ (truncated Gaussian). For the baseline, \(\mu_{n_B}=0.005\) and \(\sigma_{n_B}=0.001\), modeling an ultra–low-noise environment (e.g., nighttime operation with filtering) with a physical cap at $0.5$.
    \item \textbf{Monte Carlo samples ($K$):} $K=10^6$ i.i.d.\ draws $(\eta_i,\overline{n}_{B,i})$. This provides good empirical resolution over most of the risk range considered here. For the most stringent budgets near $10^{-5}$, however, the corresponding empirical quantiles are estimated from relatively sparse tail samples and should therefore be interpreted with the usual Monte Carlo granularity of rare-event estimation.
\end{itemize}
We evaluate risk budgets on the numerical grids described below and reuse the static sample set as described below, enabling rapid exploration of the Pareto frontier without repeated sampling.

\subsection{Simulation Procedure}
We map the risk–performance Pareto frontier with an efficient two-stage pipeline that avoids re-running expensive Monte Carlo simulations for each operating point.

In a one-time setup, we draw a large, fixed set of channel realizations and precompute the induced performance metrics:
\begin{enumerate}
    \item Generate $K = 10^6$ i.i.d.\ samples $(\eta_i,\overline{n}_{B,i})$ from the stipulated distributions, with $\eta_i$ drawn from the truncated lognormal law on $(0,1]$. Each sample represents one \emph{frame-level} channel realization, i.e., one quasi-static pair held fixed over the corresponding frame of \(n\) channel uses.
    \item For each sample, compute the covertness constant $c_{\text{cov},i}$ and the instantaneous achievable rate
    \[
        r_i \;\triangleq\; R_{\text{ach}}(\eta_i,\overline{n}_{B,i})
        \;=\; \big(1-H(\vec p(\eta_i,\overline{n}_{B,i}))\big)^+ .
    \]
    This yields two stored arrays $\{c_{\text{cov},i}\}_{i=1}^K$ and $\{r_i\}_{i=1}^K$.
\end{enumerate}

We next evaluate the risk--performance frontier by reusing the static arrays. The specific risk points are chosen based on the analysis type. For the throughput curves and the decade-gain calculations in Table~\ref{tab:risk_gain}, the symmetric risk threshold $\epsilon$ is swept over logarithmically spaced values spanning $10^{-5}$ to $10^{-1}$. For the sensitivity plot in Figure~\ref{fig:sensitivities}, we use the focused range $\epsilon\in[10^{-4},10^{-1}]$, where the estimated PDFs are numerically well behaved.\footnote{At the extreme lower end of this sweep, empirical quantiles are based on rare tail samples. The corresponding points are still useful for showing the trend toward highly conservative operation, but they should not be over-interpreted as high-precision estimates of ultra-rare-event quantiles.} To generate the 2D scheduled-payload surface, we evaluate the performance on a $20\times 20$ grid in which $\epsilon_{\text{cov}}$ and $\epsilon_{\text{rel}}$ are each drawn from the same logarithmically spaced vector spanning $10^{-5}$ to $10^{-1}$. For each risk point, we perform the following:
\begin{enumerate}
    \setcounter{enumi}{2}
    \item Estimate the required quantiles from the stored samples and set
    \[
    q_{\max} \;=\; \min\!\left\{1,\; \frac{2\delta}{\sqrt{n}}\; \widehat{F}^{-1}_{\,c_{\text{cov}}}(\epsilon_{\text{cov}})\right\},
    \qquad
    R_{\max} \;=\; \widehat{F}^{-1}_{\,R_{\text{ach}}}(\epsilon_{\text{rel}}),
    \]
    where $\widehat{F}^{-1}$ denotes the empirical quantile function.
    \item The optimal strategy is $(q^*,R^*)=(q_{\max},R_{\max})$ with per-use design throughput $T^* = q^* R^*$ and total scheduled payload \(nT^*\) over \(n\) channel uses, where the cap in $q_{\max}$ enforces the global box constraint $q\le 1$.
\end{enumerate}

All experiments use a fixed random seed for reproducibility. With $K=10^6$, the empirical quantiles are stable over the main portion of the plotted risk range. For analyses that include budgets near $10^{-5}$, the corresponding estimates are necessarily based on sparse tail samples and should be interpreted with that finite-sample resolution in mind. Simulations were implemented in MATLAB on a standard desktop. The one-time sample generation takes minutes, and each frontier point thereafter is computed in milliseconds by reusing the cached arrays. Accordingly, the qualitative trends and relative ordering of the curves are the main objects of interpretation, especially at the smallest plotted risk levels. For Figure~\ref{fig:sensitivities}, the quantities \(S_{\text{cov}}\) and \(S_{\text{rel}}\) are evaluated numerically along the symmetric-budget line \(\epsilon_{\text{cov}}=\epsilon_{\text{rel}}=\epsilon\). Specifically, for each plotted \(\epsilon\), we estimate \(S_{\text{cov}}\) by finite differences of \(T^*(\epsilon_{\text{cov}},\epsilon)\) with respect to \(\epsilon_{\text{cov}}\) at \(\epsilon_{\text{cov}}=\epsilon\), and estimate \(S_{\text{rel}}\) by finite differences of \(T^*(\epsilon,\epsilon_{\text{rel}})\) with respect to \(\epsilon_{\text{rel}}\) at \(\epsilon_{\text{rel}}=\epsilon\).

\subsection{Analytical Benchmark: The Exponential-Noise Channel}

Before presenting numerical results for the full stochastic model, we first analyze a simplified special case that admits an exact solution. This serves two purposes: (i) it provides a theoretical benchmark to validate our Monte Carlo implementation; and (ii) it offers a transparent, closed-form example that illustrates (without numerical abstraction) how the risk-constrained optimization operates under uncertainty. In a simplified model we specialize to:
\begin{itemize}
    \item \textbf{Fixed transmittance:} $\eta=\eta_0$ is known and constant.
    \item \textbf{Exponential noise:} The mean thermal photon number follows an exponential law with rate $\lambda$,
    \[
        \overline{n}_B \sim \mathrm{Exp}(\lambda), 
        \qquad f_{\overline{n}_B}(x)=\lambda e^{-\lambda x},\; x\ge 0.
    \]
\end{itemize}
Under these assumptions, all channel uncertainty is driven by $\overline{n}_B$, enabling a closed-form optimum.

\begin{lemma}[Optimal strategy for the exponential-noise channel]
\label{lem:exponential_case}
Consider a CQC link with fixed transmittance $\eta_0$ and $\overline{n}_B \sim \mathrm{Exp}(\lambda)$. The risk-constrained optimum is $(q^*,R^*)=(q_{\max},R_{\max})$ with
\[
        q_{\max} \;=\; \min\!\left\{1,\; \frac{2\delta}{\sqrt{n}}\; k \sqrt{\frac{Z^{2}-1}{4\eta_0}}\right\},
    \]
    \[
R_{\max} \;=\; \Big[\,1 - H\!\big(\vec p(\eta_0,\, -\tfrac{1}{\lambda}\ln(\epsilon_{\text{rel}}))\big)\,\Big]^+,
\]
where $k=\tfrac{\sqrt{2\eta_0}}{1-\eta_0}$ and $Z = 1 - \tfrac{2\eta_0}{\lambda}\ln(1-\epsilon_{\text{cov}})$.
\end{lemma}

\begin{proof}[Proof] See Appendix~\ref{app:proof_lemma_1}.
\end{proof}

\paragraph{Numerical illustration.}
As a concrete example, consider a high-quality link with $\eta_0=0.99$ and exponential noise with rate $\lambda=10$. For $n=10^7$ uses and risk budgets $\epsilon_{\text{cov}}=\epsilon_{\text{rel}}=0.1$, Lemma~\ref{lem:exponential_case} gives
$q_{\max}\approx 4.59\times 10^{-4}$ and $R_{\max}\approx 0.8692$, i.e., $T^*\approx 3.99\times 10^{-4}$ and about $3.99\times 10^{3}$ scheduled covert qubits over the frame. In a slightly degraded setting with $\eta_0=0.9$ (typical of challenging FSO links), the optimum becomes much more conservative: $q_{\max}\approx 4.4\times 10^{-5}$ and $R_{\max}\approx 0.2204$, yielding fewer than $10^{2}$ scheduled covert qubits per $10^7$ uses.

These calculations highlight the steep trade-off imposed by stringent covertness and reliability requirements: maintaining very small covertness-outage and decoding-failure probabilities requires extremely sparse transmissions and strong coding, which sharply reduces throughput. Nevertheless, such ultra-low-rate links can be indispensable in high-risk scenarios where the \emph{existence} of communication must remain hidden. 

The closed-form solution in Lemma~\ref{lem:exponential_case} provides a ground-truth reference for validating our Monte Carlo pipeline and clarifies how risk budgets couple with the noise-tail parameter $\lambda$ to shape the optimal strategy. To confirm correctness, we compare the Monte Carlo procedure to Lemma~\ref{lem:exponential_case} using $\eta_0=0.9$, $\lambda=10$, and $n=10^7$. Table~\ref{tab:benchmark_validation} shows close agreement across a wide range of risk thresholds, with relative errors below $5\%$ whenever $R_{\max}>0$. For $\epsilon=10^{-3}$, both methods correctly yield $R_{\max}\approx 0$, indicating that a minimum risk budget must be exceeded before a nontrivial covert link can be established in this channel.

\begin{table}[t!]
    \centering
    \small
    \caption{Validation of Monte Carlo against the analytical benchmark (fixed $\eta_0$, exponential $\overline{n}_B$) from Lemma~\ref{lem:exponential_case}, under symmetric risk budgets \(\epsilon_{\text{cov}}=\epsilon_{\text{rel}}=\epsilon\). The reported percentage errors are computed from full-precision values, whereas the displayed theory and simulation entries are rounded for readability.}
    \label{tab:benchmark_validation}
    \begin{tabular}{l c c c c}
        \toprule
        \textbf{Risk ($\epsilon$)} & \textbf{Metric} & \textbf{Theory} & \textbf{MC Sim.} & \textbf{Error (\%)} \\
        \midrule
        \multirow{2}{*}{$10^{-3}$} & $q_{\max}$ & $4.0 \times 10^{-6}$ & $4.0 \times 10^{-6}$ & 1.15 \\
                                   & $R_{\max}$ & $\approx 0$ & $\approx 0$ & --- \\
        \midrule
        \multirow{2}{*}{$10^{-2}$} & $q_{\max}$ & $1.3 \times 10^{-5}$ & $1.3 \times 10^{-5}$ & 0.49 \\
                                   & $R_{\max}$ & 0.011262 & 0.011724 & 4.10 \\
        \midrule
        \multirow{2}{*}{$10^{-1}$} & $q_{\max}$ & $4.4 \times 10^{-5}$ & $4.4 \times 10^{-5}$ & 0.04 \\
                                   & $R_{\max}$ & 0.220380 & 0.220336 & 0.02 \\
        \midrule
        \multirow{2}{*}{$0.2$}     & $q_{\max}$ & $6.4 \times 10^{-5}$ & $6.4 \times 10^{-5}$ & 0.06 \\
                                   & $R_{\max}$ & 0.294959 & 0.294767 & 0.07 \\
        \midrule
        \multirow{2}{*}{$0.5$}     & $q_{\max}$ & $1.15 \times 10^{-4}$ & $1.15 \times 10^{-4}$ & 0.04 \\
                                   & $R_{\max}$ & 0.404270 & 0.404337 & 0.02 \\
        \bottomrule
    \end{tabular}
\end{table}

\subsection{Results and Interpretation}

Before turning to the numerical plots, we formalize two structural properties of the trade-offs we observe: monotonicity of the performance frontier (Corollary~\ref{cor:monotonicity}) and the sensitivities of the optimum to the risk budgets (Proposition~\ref{prop:sensitivity}).
Throughout this section, the quantities \(T^*=q^*R^*\) and \(nT^*\) should be interpreted as \emph{scheduled design payloads} under the chosen operating point, not as outage-averaged realized goodput.

\begin{corollary}[Monotonicity of the Pareto frontier]
\label{cor:monotonicity}
For the risk-constrained problem, the optimal per-use throughput
\(
T^*(\epsilon_{\text{cov}},\epsilon_{\text{rel}})=q^* R^*
\)
is monotonically nondecreasing in each risk budget $\epsilon_{\text{cov}}$ and $\epsilon_{\text{rel}}$.
\end{corollary}

\begin{proof}
Let
\[
S(\epsilon_{\text{cov}},\epsilon_{\text{rel}})
=
[0,q_{\max}(\epsilon_{\text{cov}})]\times[0,R_{\max}(\epsilon_{\text{rel}})]
\]
denote the feasible rectangle induced by the two risk budgets. Since quantiles are monotonically nondecreasing in their probability level, both \(q_{\max}(\epsilon_{\text{cov}})\) and \(R_{\max}(\epsilon_{\text{rel}})\) are monotonically nondecreasing in their respective arguments.

Therefore, if \(\epsilon_{\text{cov},1}\le \epsilon_{\text{cov},2}\) while \(\epsilon_{\text{rel}}\) is fixed, then
\[
S(\epsilon_{\text{cov},1},\epsilon_{\text{rel}})
\subseteq
S(\epsilon_{\text{cov},2},\epsilon_{\text{rel}}).
\]
Likewise, if \(\epsilon_{\text{rel},1}\le \epsilon_{\text{rel},2}\) while \(\epsilon_{\text{cov}}\) is fixed, then
\[
S(\epsilon_{\text{cov}},\epsilon_{\text{rel},1})
\subseteq
S(\epsilon_{\text{cov}},\epsilon_{\text{rel},2}).
\]

Therefore,
\[
\max_{(q,R)\in S(\epsilon_{\text{cov},1},\epsilon_{\text{rel}})} qR
\;\le\;
\max_{(q,R)\in S(\epsilon_{\text{cov},2},\epsilon_{\text{rel}})} qR,
\]
and similarly when only \(\epsilon_{\text{rel}}\) increases.

Hence \(T^*(\epsilon_{\text{cov}},\epsilon_{\text{rel}})\) is monotonically nondecreasing in each risk budget.
\end{proof}

To identify which risk constraint is the active bottleneck, we define risk sensitivities.

\begin{definition}[Risk sensitivity]
The covertness and reliability risk sensitivities are the partial derivatives of the optimal throughput with respect to their risk budgets:
\[
S_{\text{cov}} \;\triangleq\; \frac{\partial T^*}{\partial \epsilon_{\text{cov}}}, 
\qquad
S_{\text{rel}} \;\triangleq\; \frac{\partial T^*}{\partial \epsilon_{\text{rel}}}.
\]
\end{definition}

\begin{proposition}[Sensitivity formulas away from the saturation point]
\label{prop:sensitivity}
Assume that the evaluated quantiles \(F^{-1}_{c_{\text{cov}}}(\epsilon_{\text{cov}})\) and \(F^{-1}_{R_{\text{ach}}}(\epsilon_{\text{rel}})\) lie at points where the corresponding CDFs are differentiable and the associated densities are strictly positive. Define
\[
\widetilde q(\epsilon_{\text{cov}})
\;\triangleq\;
\frac{2\delta}{\sqrt{n}}\,F^{-1}_{c_{\text{cov}}}(\epsilon_{\text{cov}}),
\qquad
q_{\max}(\epsilon_{\text{cov}})=\min\{1,\widetilde q(\epsilon_{\text{cov}})\}.
\]
Then, away from the transition point \(\widetilde q(\epsilon_{\text{cov}})=1\),
\[
S_{\text{cov}}
=
\begin{cases}
\displaystyle
\frac{2\delta}{\sqrt{n}}\,
\frac{R_{\max}}{\,f_{c_{\text{cov}}}\!\big(F^{-1}_{c_{\text{cov}}}(\epsilon_{\text{cov}})\big)},
& \text{if } \widetilde q(\epsilon_{\text{cov}})<1,\\[2.2ex]
0, & \text{if } \widetilde q(\epsilon_{\text{cov}})>1,
\end{cases}
\]
and
\[
S_{\text{rel}}
= \frac{q_{\max}}{\,f_{R_{\text{ach}}}\!\big(F^{-1}_{R_{\text{ach}}}(\epsilon_{\text{rel}})\big)}\,.
\]
\end{proposition}

\begin{proof}
Since \(T^*=q_{\max}R_{\max}\) and \(q_{\max}\) does not depend on \(\epsilon_{\text{rel}}\), the derivative with respect to \(\epsilon_{\text{rel}}\) follows from the chain rule and the inverse-function derivative
\[
\frac{d}{d\alpha}F^{-1}_X(\alpha)=\big[f_X(F^{-1}_X(\alpha))\big]^{-1},
\]
valid at quantile points where the CDF is differentiable and the density is strictly positive. For \(\epsilon_{\text{cov}}\), write \(q_{\max}=\min\{1,\widetilde q(\epsilon_{\text{cov}})\}\). If \(\widetilde q(\epsilon_{\text{cov}})<1\), differentiate \(\widetilde q\) directly; if \(\widetilde q(\epsilon_{\text{cov}})>1\), then \(q_{\max}=1\) is locally constant and hence \(S_{\text{cov}}=0\).
\end{proof}

\noindent\textit{Remark.} At the transition point \(\widetilde q(\epsilon_{\text{cov}})=1\), the sensitivity \(S_{\text{cov}}\) need not exist in the classical sense because the minimum operator creates a kink, so one should interpret it using one-sided derivatives. Moreover, because \(R_{\text{ach}}(\eta,\overline{n}_B)=\big(1-H(\vec p(\eta,\overline{n}_B))\big)^+\) can induce an atom at \(R_{\text{ach}}=0\), the sensitivities in Proposition~\ref{prop:sensitivity} need not exist when the relevant quantile falls on that atom or on another nondifferentiable point. In such cases, one should instead use one-sided finite differences or a subgradient-style interpretation. In the numerical results below, we therefore estimate these sensitivities from the computed risk-constrained frontier using finite differences, rather than by direct density estimation.

\begin{figure}[t]
    \centering
    \includegraphics[width=0.5\linewidth]{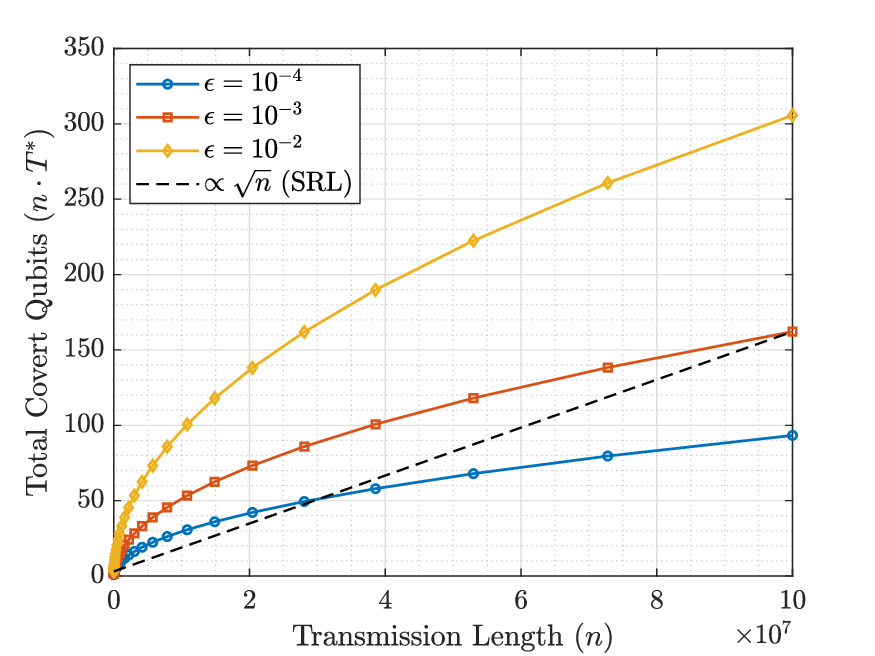}
    \caption{Scaling of total scheduled covert qubits ($nT^*$) with transmission length $n$ for several fixed symmetric risk budgets $\epsilon=\epsilon_{\mathrm{cov}}=\epsilon_{\mathrm{rel}}$. The dashed line indicates the $\sqrt{n}$ Square-Root-Law scaling.}
    \label{fig:pareto}
\end{figure}

Figure~\ref{fig:pareto} shows the primary risk-constrained result: the scaling of total scheduled covert qubits $n\,T^*$ with transmission length $n$ for several fixed symmetric risk budgets $\epsilon=\epsilon_{\text{cov}}=\epsilon_{\text{rel}}$. For any fixed $n$, larger symmetric risk budgets yield larger total scheduled covert payloads, which is consistent with Corollary~\ref{cor:monotonicity}.

As an illustrative operating point, take $\epsilon_{\text{cov}}=\epsilon_{\text{rel}}=0.01$ and $n=10^7$. For the baseline channel, the optimizer yields a sparse-transmission operating point with total scheduled payload on the order of $10^2$ covert qubits over the frame, consistent with Figure~\ref{fig:pareto}. We therefore avoid quoting a specific pair $(q^*,R^*)$ here, since such values should be reported only from the exact cached-sample run used to generate the plotted frontier.

Reading Figure~\ref{fig:pareto} across curves at a fixed $n$ reveals a highly nonlinear risk--reward trade-off. Quantitatively, Table~\ref{tab:risk_gain} reports the multiplicative gain in throughput when relaxing $\epsilon$ by one decade for $n=10^7$.

\begin{table}[b!]
    \centering
    \caption{Multiplicative throughput gain when relaxing risk by one decade ($n=10^7$).}
    \label{tab:risk_gain}
    \begin{tabular}{cc}
        \toprule
        \textbf{Risk relaxation} & \textbf{Throughput gain} \\
        \midrule
        $10^{-5} \to 10^{-4}$ & $2.04\times$ \\
        $10^{-4} \to 10^{-3}$ & $1.74\times$ \\
        $10^{-3} \to 10^{-2}$ & $1.89\times$ \\
        $10^{-2} \to 10^{-1}$ & $2.52\times$ \\
        \bottomrule
    \end{tabular}
\end{table}

A clear takeaway is that modest relaxations from extremely conservative budgets (e.g., $10^{-5}\!\to\!10^{-4}$) can more than \emph{double} throughput, highlighting the performance cost of worst-case-like operation.

Moving along any single curve, where $\epsilon$ is fixed and $n$ increases, isolates the effect of transmission length. In the uncapped regime relevant to Figure~\ref{fig:pareto}, the Square-Root Law implies $q_{\max}\propto 1/\sqrt{n}$ while $R_{\max}$ is independent of $n$, so $T^*=q^*R^*\propto 1/\sqrt{n}$ and therefore $n\,T^*$ grows only as $\sqrt{n}$. Thus, longer transmissions yield more total scheduled covert qubits but with diminishing returns. Under a per-session interpretation of the risk budgets, this also suggests that multiple shorter sessions can be preferable to a single prolonged transmission when those sessions can be treated as operationally separate.

\subsection{Impact of Channel Quality and Volatility}

\begin{figure}[t!]
    \centering
    \includegraphics[width=0.5\linewidth]{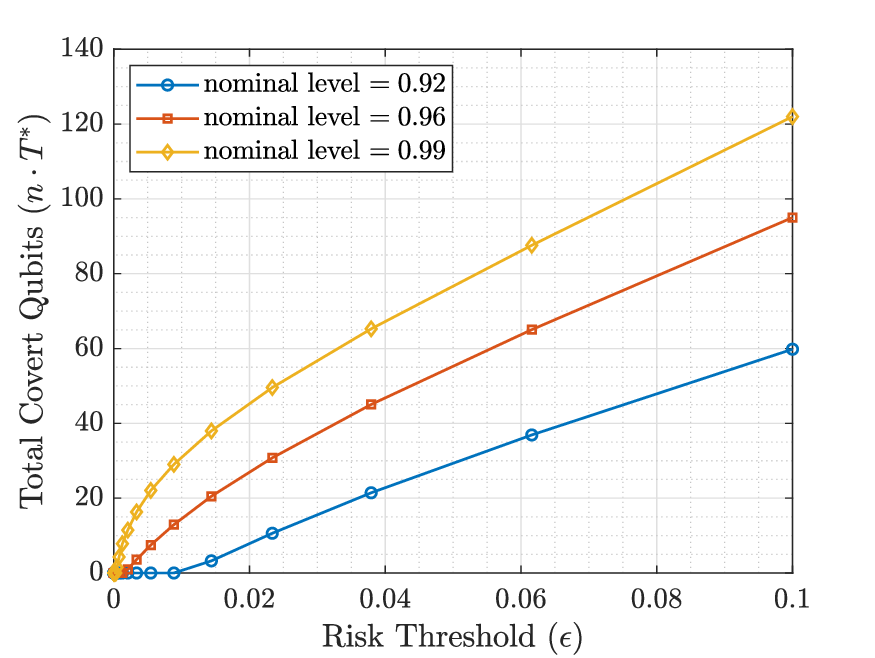}
    \caption{Impact of the nominal channel-quality scenario on the symmetric-risk trade-off curve for $n=10^7$, with $\epsilon_{\text{cov}}=\epsilon_{\text{rel}}=\epsilon$. The three curves correspond to jointly specified \emph{poor}, \emph{nominal}, and \emph{excellent} transmittance-and-noise regimes. Better nominal channel conditions offer a better risk-performance trade-off, while the poor regime shows a feasibility boundary, requiring a risk threshold of $\epsilon>0.01$.}
    \label{fig:pareto_mean_eta}
\end{figure}

We next quantify how physical channel characteristics shape achievable performance. Beyond the baseline scenario described above, we also study comparative channel-quality scenarios to illustrate how the frontier changes across different jointly specified transmittance-and-noise regimes. Figure~\ref{fig:pareto_mean_eta} compares three channel-quality scenarios (\emph{poor}, \emph{nominal}, and \emph{excellent}) defined as follows: the transmittance is modeled as a truncated lognormal law on $(0,1]$ with fixed $\sigma_{\ln(\eta)}=0.08$ and location parameter $\mu_{\ln(\eta)}=\log(m)-\tfrac{1}{2}(0.08)^2$, where the nominal target levels are $m\in\{0.92,\,0.96,\,0.99\}$ for the \emph{poor}, \emph{nominal}, and \emph{excellent} cases, respectively. Because the distribution is truncated to $(0,1]$, these values should be interpreted as scenario targets (and plot labels) rather than exact post-truncation means. The corresponding truncated-Gaussian noise models are $\overline{n}_B\sim\mathcal{N}_{[0,0.5]}(0.02,0.01^2)$, $\mathcal{N}_{[0,0.5]}(0.01,0.005^2)$, and $\mathcal{N}_{[0,0.5]}(0.005,0.001^2)$. These curves provide a quantitative answer to how much a jointly better transmittance-and-noise regime is worth under the adopted scenario family.

The plot also reveals a feasibility boundary for the poorer regime. For excellent and nominal channels, a viable covert link exists even at very small risks. For the poor channel, throughput remains zero until $\epsilon \gtrsim 0.01$. Thus, for low-quality channels, a minimum level of operational risk must be accepted to realize a nontrivial covert link, an important planning insight for challenging environments.

\begin{figure}[b!]
    \centering
    \includegraphics[width=0.5\linewidth]{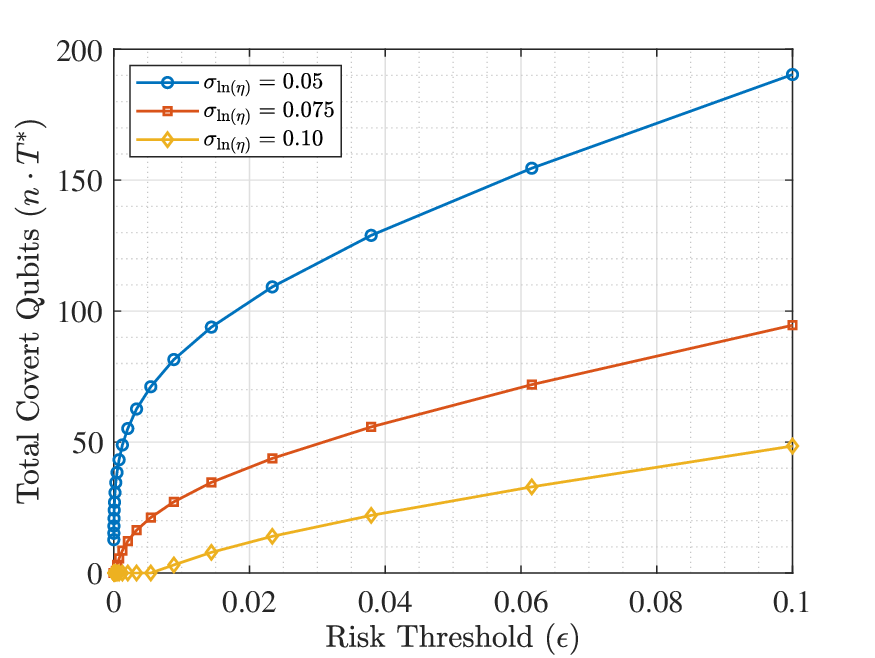}
    \caption{Effect of transmittance volatility on the symmetric-risk trade-off curve at fixed transmission length $n=10^7$, with $\epsilon_{\text{cov}}=\epsilon_{\text{rel}}=\epsilon$. The three curves correspond to $\sigma_{\ln(\eta)} \in \{0.05,\,0.075,\,0.10\}$ under the same noise model, with $\mu_{\ln(\eta)}$ adjusted across cases so that the mean transmittance remains approximately fixed. Increased volatility markedly degrades the risk--performance trade-off, and the most volatile case exhibits a low-risk feasibility boundary.}
    \label{fig:pareto_volatility}
\end{figure}

We next examine the role of channel \emph{predictability} (volatility). Figure~\ref{fig:pareto_volatility} compares symmetric-risk trade-off curves for three channels under the same noise model and approximately the same mean transmittance, but with different volatility levels, parameterized by $\sigma_{\ln(\eta)} \in \{0.05,\,0.075,\,0.10\}$. For each volatility level, $\mu_{\ln(\eta)}$ is adjusted so that the resulting transmittance distribution remains centered at approximately the same mean level. Thus, the observed differences are attributable primarily to volatility.

Performance degrades sharply as volatility increases. For example, at $\epsilon \approx 1.4\times 10^{-2}$, the most stable case ($\sigma_{\ln(\eta)}=0.05$) supports about $94$ total scheduled covert qubits, compared with about $33$ for $\sigma_{\ln(\eta)}=0.075$ and only about $8$ for $\sigma_{\ln(\eta)}=0.10$. Thus, increasing the volatility from $0.05$ to $0.10$ reduces throughput by more than an order of magnitude in this operating region. The most volatile case also shows a clear feasibility boundary at stringent risk levels: over the plotted grid, it remains effectively infeasible until the risk budget is relaxed to around $10^{-2}$, after which nonzero throughput becomes possible. These results show that performance depends not only on average channel quality but also critically on predictability.

\begin{figure}[t!]
    \centering
    \includegraphics[width=0.5\linewidth]{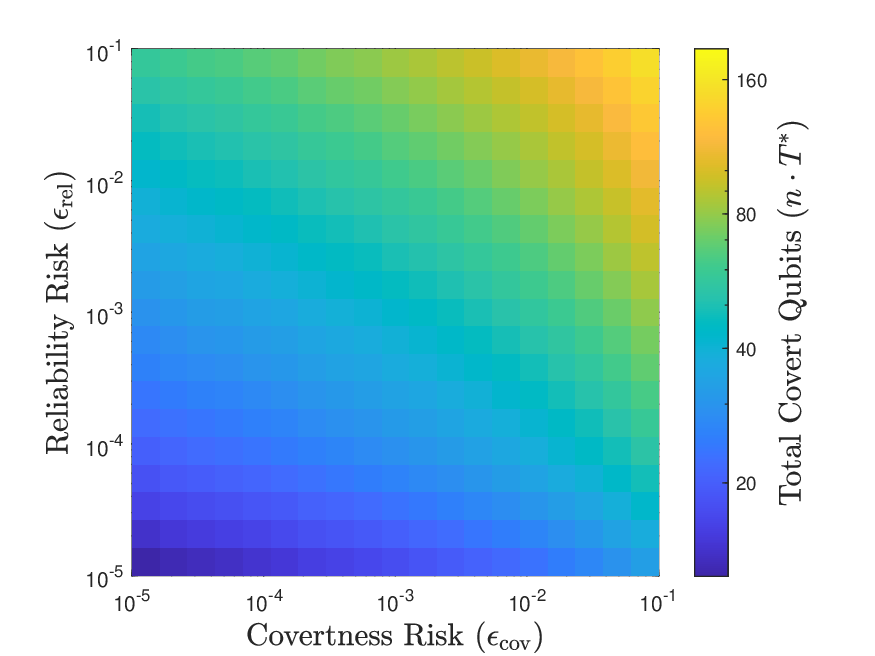}
    \caption{Surface of total scheduled covert qubits ($n\,T^*$) as a function of independent covertness and reliability risks $(\epsilon_{\text{cov}},\epsilon_{\text{rel}})$ for $n=10^7$. Over the plotted grid, the surface appears visually smooth and increases in both risk budgets; however, visual symmetry on this plot should not be interpreted as equal local sensitivity.}
    \label{fig:rate_surface}
    \vspace{-4.5mm}
\end{figure}

Figure~\ref{fig:rate_surface} separates the two risks and shows the full scheduled-payload surface for $n=10^7$. For our high-quality channel, the plotted surface is visually smooth and monotone in both $\epsilon_{\text{cov}}$ and $\epsilon_{\text{rel}}$, confirming that relaxing either risk budget improves the achievable scheduled payload. Visually, the surface does not exhibit a sharp cliff in either direction over the plotted range, which is consistent with well-behaved underlying distributions for $c_{\text{cov}}$ and $R_{\text{ach}}$ in this regime. However, this visual smoothness should not be over-interpreted as equal local sensitivity to the two risks.

For the smooth high-quality baseline regime considered here, where the cap in \(q_{\max}=\min\{1,\widetilde q(\epsilon_{\text{cov}})\}\) is inactive (that is, \(\widetilde q(\epsilon_{\text{cov}})<1\) over the plotted range), Proposition~\ref{prop:sensitivity} shows that the relevant local quantities are the partial derivatives $S_{\text{cov}}$ and $S_{\text{rel}}$, which depend on the corresponding quantile slopes of $c_{\text{cov}}$ and $R_{\text{ach}}$.
Figure~\ref{fig:sensitivities} reveals that, for the present high-quality channel, $S_{\text{cov}} \gg S_{\text{rel}}$, so covertness remains the dominant bottleneck even though the 2D surface looks qualitatively smooth in both directions. In lower-quality or more volatile channels, we expect this imbalance to change, potentially making reliability the dominant bottleneck near a sharp payload-collapse regime.

\begin{figure}
    \centering
    \includegraphics[width=0.5\linewidth]{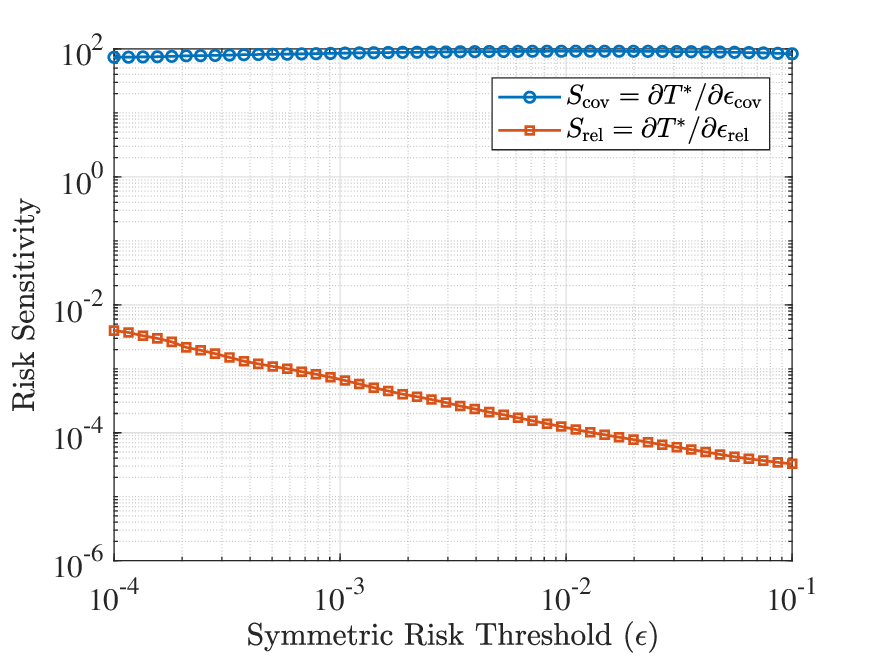}
    \caption{Sensitivity along the symmetric-budget line $\epsilon_{\text{cov}}=\epsilon_{\text{rel}}=\epsilon$ for the high-quality baseline channel ($\sigma_{\ln(\eta)}=0.05$, $\mu_{n_B}=0.005$, $\sigma_{n_B}=0.001$; $n=10^7$). The covertness sensitivity $S_{\text{cov}}$ dominates $S_{\text{rel}}$.}
    \label{fig:sensitivities}
\end{figure}

Finally, Figure~\ref{fig:sensitivities} plots the numerically estimated sensitivities $S_{\text{cov}}$ and $S_{\text{rel}}$ versus symmetric $\epsilon$. Over the plotted risk range, $S_{\text{cov}}$ exceeds $S_{\text{rel}}$ by several orders of magnitude, confirming that covertness is the active bottleneck for this high-quality baseline channel. Reliability is consistently high and thus less sensitive. By contrast, covertness remains fundamentally constrained by the Square-Root Law, making performance highly sensitive to the exact choice of $\epsilon_{\text{cov}}$ in the tail. For high-quality links, fine-tuning the covertness risk budget is therefore the most effective lever for improving performance.

\subsection{Exploratory Results for the Risk-Adjusted Model}

As a secondary exploratory study, we visualize the behavior of the risk-adjusted formulation by numerically maximizing the weighted objective over a fixed $(q,R)$ grid while sweeping the risk-aversion parameters $\lambda_{\text{cov}}$ and $\lambda_{\text{rel}}$. For these risk-adjusted results, we fix the $(q,R)$ search grid and the $(\lambda_{\text{cov}},\lambda_{\text{rel}})$ sweep ranges in advance, as detailed below. Figures~\ref{fig:decoupling_plots} and~\ref{fig:risk_adjusted_heatmaps} use the same representative channel setting, namely $\sigma_{\ln(\eta)}=0.07$, $\mu_{\ln(\eta)}=\log(0.96)-0.5(0.07)^2$, $\mu_{n_B}=0.01$, $\sigma_{n_B}=0.005$, $n=10^3$, $\delta=0.05$, and $K=10^6$, using a fixed random seed for reproducibility. In both cases, the weighted objective is maximized over a uniform $401\times 401$ grid on $(q,R)\in[0,1]\times[0,1]$, with $q$ and $R$ discretized as \texttt{linspace(0,1,401)}. For Figure~\ref{fig:decoupling_plots}, the one-dimensional sweeps use $40$ logarithmically spaced points from $10^{-2}$ to $10^{6}$, i.e., \texttt{logspace(-2,6,40)}. Specifically, Figure~\ref{fig:decoupling_vs_cov} sweeps $\lambda_{\text{cov}}$ with $\lambda_{\text{rel}}=1$ fixed, and Figure~\ref{fig:decoupling_vs_rel} sweeps $\lambda_{\text{rel}}$ with $\lambda_{\text{cov}}=10$ fixed. For Figure~\ref{fig:risk_adjusted_heatmaps}, the two-dimensional heatmaps use $\lambda_{\text{cov}},\lambda_{\text{rel}}\in\texttt{logspace(-6,6,25)}$. If multiple grid points attain the same maximum within numerical tolerance, MATLAB's \texttt{max(J(:))} selects the first maximizer in column-major order, which corresponds to the smallest $q$ and, within that $q$, the smallest $R$ among tied grid points. Figure~\ref{fig:decoupling_plots} shows representative sweeps of the numerically selected operating point, in which the dominant changes occur through sharp regime transitions and boundary-pinned solutions rather than through a uniform decoupling across the entire parameter range.

Such aggressive points lie outside the small-$q$ regime underlying the covertness approximation in Section~\ref{sec:model} and should therefore be read as illustrating the geometry of the weighted objective, not as deployment-ready covert operating points. Once $\lambda_{\text{cov}}$ crosses a narrow transition region, the numerically selected operating point collapses rapidly to a no-transmission regime, driving both $q^*$ and $R^*$ to zero.

The right panel (Figure~\ref{fig:decoupling_vs_rel}) shows a different behavior. With $\lambda_{\text{cov}}$ fixed, increasing the reliability penalty mainly suppresses the code rate $R^*$, which decreases steadily toward zero. Over this same sweep, $q^*$ remains close to zero throughout, indicating that in this representative operating regime the covertness penalty has already pushed the system into a highly conservative transmission state before the reliability penalty is varied.

\begin{figure}[t!]
    \centering
    \begin{subfigure}[b]{0.48\linewidth}
        \centering
        \includegraphics[width=\linewidth]{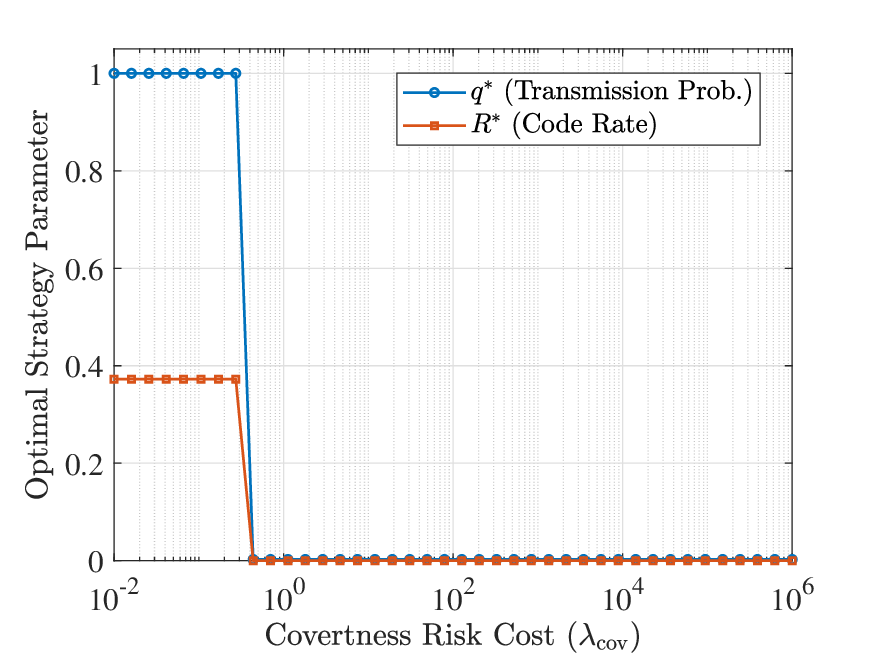}
        \caption{Sweep over $\lambda_{\text{cov}}$ (fixed $\lambda_{\text{rel}}=1$).}
        \label{fig:decoupling_vs_cov}
    \end{subfigure}
    \hfill
    \begin{subfigure}[b]{0.48\linewidth}
        \centering
        \includegraphics[width=\linewidth]{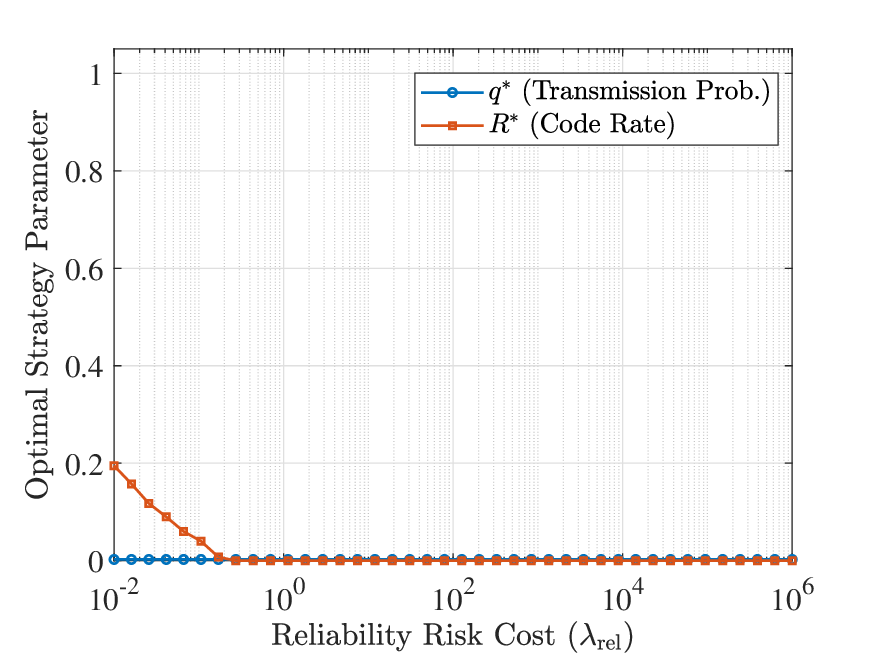}
        \caption{Sweep over $\lambda_{\text{rel}}$ (fixed $\lambda_{\text{cov}}=10$).}
        \label{fig:decoupling_vs_rel}
    \end{subfigure}
    \caption{Exploratory sweeps in the risk-adjusted model. The plotted operating points are numerical maximizers of the weighted objective on the fixed grid described in Section~\ref{sec:results}; they are useful for illustrating regime changes in the weighted formulation, but they should not be interpreted as guarantee-bearing covert operating points. (a) As the covertness penalty increases, the numerically selected operating point undergoes a sharp transition from an aggressive operating point to a no-transmission regime. (b) With the covertness penalty fixed, increasing the reliability penalty mainly suppresses the code rate, while the transmission probability remains pinned near zero in this operating regime.}
    \label{fig:decoupling_plots}
\end{figure}

The broader optimization landscape in Figure~\ref{fig:risk_adjusted_heatmaps} confirms that the numerically selected risk-adjusted operating point is governed by sharp regime changes rather than by a uniformly smooth decoupling. In Figure~\ref{fig:q_star_heatmap}, the dominant feature is a covertness-driven transition boundary: for sufficiently small $\lambda_{\text{cov}}$, the optimizer selects an aggressive transmission probability, whereas beyond a critical region it collapses rapidly to a no-transmission regime with $q^*\approx 0$. The location of this boundary is only weakly affected by $\lambda_{\text{rel}}$ until the reliability penalty becomes very large.

Figure~\ref{fig:r_star_heatmap} shows a complementary but not fully separable behavior. For small reliability penalties, the optimizer keeps a high code rate, whereas increasing $\lambda_{\text{rel}}$ progressively compresses $R^*$ toward zero. However, the transition is not purely horizontal: its detailed shape depends on $\lambda_{\text{cov}}$, especially near the regime boundary where the numerically selected operating point is already close to the no-transmission state. Thus, the risk-adjusted formulation does provide useful policy knobs, but their effects are strongly operating-point dependent and are better interpreted as regime-dependent control rather than as a globally separable law.

\begin{figure}[t!]
    \centering
    \begin{subfigure}[b]{0.48\linewidth}
        \centering
        \includegraphics[width=\linewidth]{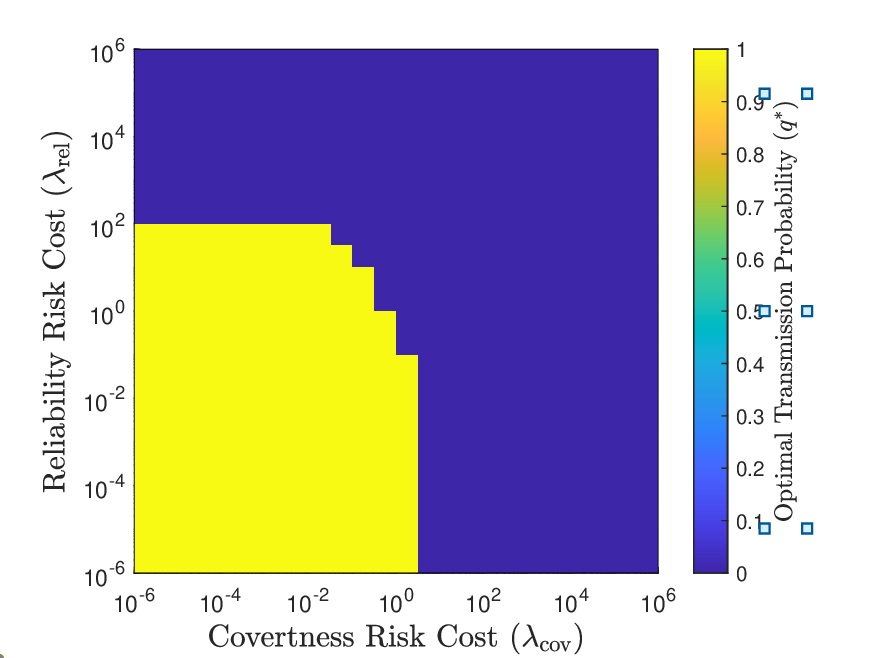}
        \caption{Optimal transmission probability $q^*$.}
        \label{fig:q_star_heatmap}
    \end{subfigure}
    \hfill
    \begin{subfigure}[b]{0.48\linewidth}
        \centering
        \includegraphics[width=\linewidth]{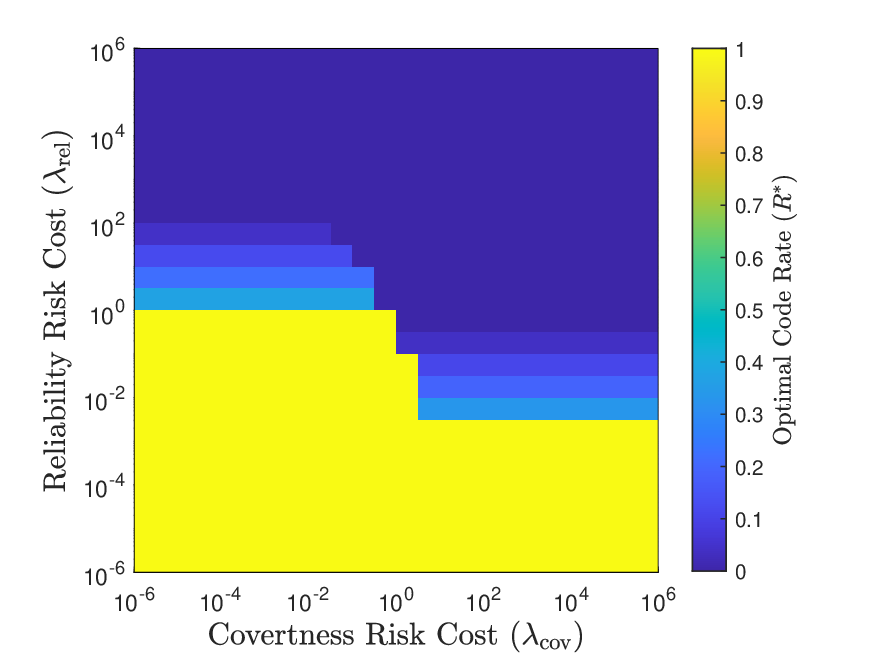}
        \caption{Optimal code rate $R^*$.}
        \label{fig:r_star_heatmap}
    \end{subfigure}
    \caption{Exploratory heatmaps of the numerically selected operating point over $(\lambda_{\text{cov}},\lambda_{\text{rel}})$, obtained by direct maximization of the weighted objective on the fixed $(q,R)$ grid described in Section~\ref{sec:results}. These plots illustrate the geometry of the weighted objective rather than guarantee-bearing covert operating points. The dominant changes occur across regime boundaries: $q^*$ undergoes a sharp covertness-driven transition between aggressive transmission and a no-transmission regime, while $R^*$ is strongly compressed as the reliability penalty increases, with the precise transition structure depending on the operating point.}
    \label{fig:risk_adjusted_heatmaps}
\end{figure}

\section{Discussion}
\label{sec:discussion}

We interpret our results in a broader systems context, extract engineering lessons, and highlight limitations that point to future research.

\subsection{Our Proposed Risk-Aware Design}
Existing analyses of CQC, such as~\cite{anderson2024square}, assume perfectly known, static channels. These models provide theoretical benchmarks but rarely reflect practice. Our framework instead models channel parameters $(\eta, \overline{n}_{B})$ as random variables, yielding a \emph{risk--performance Pareto frontier} that captures the trade-off between throughput and covertness-outage and decoding-failure probabilities. This shift from deterministic to stochastic modeling enables a paradigm of probabilistic assurance: within the adopted stochastic model and over the risk range studied here, relaxing extremely stringent risk budgets to small, explicitly quantified outage levels can unlock more-than-one-order-of-magnitude throughput gains relative to very conservative operating points. As channel uncertainty becomes small in the sense that the distributions of $\eta$ and $\overline{n}_B$ collapse to deterministic values, the induced distributions of $c_{\text{cov}}(\eta,\overline{n}_B)$ and $R_{\text{ach}}(\eta,\overline{n}_B)$ likewise collapse to deterministic values, and the present framework correspondingly reduces to the associated deterministic benchmark of prior foundational models~\cite{anderson2024square}.

\subsection{Design Insights and Engineering Principles}
Our results yield several actionable guidelines:
\begin{itemize}
    \item \textbf{Risk Budgets as System Resources.} Outage probabilities $(\epsilon_{\text{cov}}, \epsilon_{\text{rel}})$ can be treated like power or bandwidth. Allocating them according to mission priorities enables principled, data-driven operating points rather than heuristic tuning, as demonstrated by the explicit risk-reward trade-offs mapped in the Pareto frontiers of Figure~\ref{fig:pareto} and Figure~\ref{fig:pareto_mean_eta}.

    \item \textbf{Operating-Point-Dependent Control in the Risk-Adjusted Model.} The numerically selected operating point in the risk-adjusted formulation is influenced differently by the two penalties, but the effect is regime dependent rather than globally separable. In our representative operating point, increasing the covertness penalty mainly drives a sharp transition in $q^{*}$ toward a no-transmission regime, while increasing the reliability penalty mainly compresses $R^{*}$. Figures~\ref{fig:decoupling_plots} and~\ref{fig:risk_adjusted_heatmaps} therefore suggest useful tuning knobs, but they should be interpreted as boundary- and operating-point-dependent controls rather than as a universally decoupled law. Accordingly, this formulation is best viewed as a policy-exploration tool when designers wish to trade throughput against weighted failure costs, whereas the risk-constrained model remains the appropriate choice when strict outage guarantees are required.

    \item \textbf{Constraint Dominance Reveals Bottlenecks.} In the high-quality baseline channel studied here, covertness dominates performance; relaxing $\epsilon_{\text{cov}}$ delivers the greatest gains. In lower-quality or more volatile channels, this balance can shift, potentially making reliability the dominant bottleneck near a payload-collapse regime. Identifying the active constraint is a key to effective resource allocation.

    \item \textbf{Volatility Matters as Much as Averages.} Even when the mean transmittance is held approximately fixed, increased variance in transmittance can sharply compress the feasible operating region, especially at stringent risk budgets. As shown in Figure~\ref{fig:pareto_volatility}, stabilizing the channel (e.g., adaptive optics or scheduling in favorable atmospheric windows) can be as important as improving the average channel quality.

    \item \textbf{Value of Sensing.} Better channel characterization can enlarge the safe operating region by reducing uncertainty in the governing channel laws. In particular, Figure~\ref{fig:pareto_volatility} shows directly that lower transmittance volatility yields a strictly superior performance frontier, justifying investment in environmental sensing and channel-state prediction.

    \item \textbf{Small-$q$ Covertness Approximation.} The covertness constraint used in our analysis is derived from a low-transmission-probability approximation. In the risk-constrained formulation, the optimal strategies remain in the sparse-transmission regime for the parameter ranges studied here. In the risk-adjusted formulation, however, very small penalties can produce mathematically optimal points with large $q$, which should be interpreted only qualitatively.
\end{itemize}

\subsection{Model Limitations}
Our formulation relies on simplifying assumptions:
\begin{itemize}
    \item \textbf{Stationarity.} We assume fixed distributions for $\eta$ and $\overline{n}_{B}$; in reality, channels can be non-stationary. The model is most applicable within a single coherence interval.
    \item \textbf{Non-Adaptive Transmission.} The strategy $(q,R)$ is fixed. Adaptive schemes with feedback could outperform this baseline but fall outside our scope.
    \item \textbf{Passive Adversary.} Willie is assumed to be powerful but passive. Active adversaries capable of injecting or probing would require game-theoretic extensions.
    \item \textbf{Distribution Misspecification.} The framework assumes that the governing distributions of \(\eta\) and \(\overline{n}_B\) are known well enough to estimate the relevant quantiles. If these distributions are misspecified, the nominal risk budgets may be miscalibrated, so the actual covertness and reliability outage probabilities can differ from the designed values.
    \item \textbf{Quasi-Static Frame Assumption.} The stochastic model treats \((\eta,\overline{n}_B)\) as constant within each frame and random only across frames/sessions. If the channel varies substantially within a frame, then the present outage formulation should be replaced by a finer time-varying model.
\end{itemize}

\subsection{Future Research}
Several extensions are needed before this framework can support deployment-oriented design:
\begin{itemize}
    \item \textbf{Adaptive and Learning-Based Strategies.} Incorporating partial feedback or online estimation to adapt $(q_t,R_t)$ dynamically.
    \item \textbf{Correlated and Time-Varying Channels.} Using copulas or time-series models to capture dependence between $\eta$ and $\overline{n}_{B}$ and to address non-stationarity.
    \item \textbf{Robustness Against Active Wardens.} Extending to adversaries who probe or interfere, requiring robust game-theoretic designs.
    \item \textbf{Distributionally Robust Risk Design.} Extending the framework to ambiguity sets or distributionally robust formulations would allow one to hedge against errors in the assumed laws of \(\eta\) and \(\overline{n}_B\), thereby separating stochastic uncertainty from statistical model uncertainty.
\end{itemize}

\section{Conclusion}
\label{sec:conclusion}

This paper introduces a risk-aware framework for covert quantum communication (CQC) that addresses realistic channel uncertainty. By modeling transmittance and thermal noise as random variables, we formulate a primary risk-constrained design that maximizes throughput while explicitly bounding covertness and reliability outage probabilities. We also include a secondary risk-adjusted extension to illustrate how weighted failure costs reshape preferred operating points, but we do not position that extension as the guarantee-bearing design.

Our Monte Carlo-based methodology, validated against a new analytical benchmark, characterizes the risk-performance trade-off frontier under the adopted stochastic model. We demonstrate that, under the same adopted stochastic model and over the risk range studied here, relaxing extremely stringent outage budgets to marginal but quantified risk levels can increase covert throughput by more than one order of magnitude relative to very conservative operating points. Viewed through a security lens, this means that the paper provides a way to calibrate the probabilities of communication exposure and decoding failure under stochastic uncertainty, rather than relying on nominal or worst-case assumptions alone. The analysis reveals key engineering principles, notably that the risk-adjusted design is governed by operating-point-dependent regime changes: covertness and reliability penalties act as useful tuning knobs, but their effects are not globally separable. By establishing these principles of risk-aware design, this work provides a concrete analytical and simulation-based tool for tuning CQC operating points for specific missions under the adopted model. More broadly, the framework offers a foundation for future work on covert quantum communication under realistic uncertainty, including adaptive, correlated, deployment-oriented, and distributionally robust extensions.

\bibliographystyle{ACM-Reference-Format}
\bibliography{references}

\appendix
\section{Proof of Theorem~\ref{thm:optimal_throughput} (Optimal risk-constrained throughput)}
\label{Proof:thm_risk_constrained}

Consider the risk-constrained program \eqref{eq:objective_constrained}–\eqref{eq:rel_outage} with objective \(T(q,R)=qR\), probabilistic constraints
\[
\mathbb{P}\!\left[q > \tfrac{2\delta}{\sqrt{n}}\,c_{\text{cov}}(\eta,\overline{n}_B)\right] \le \epsilon_{\text{cov}},
\qquad
\mathbb{P}\!\left[R > R_{\text{ach}}(\eta,\overline{n}_B)\right] \le \epsilon_{\text{rel}},
\]
and box constraints \(0\le q\le 1\), \(0\le R\le 1\). Write
\[
c_{\text{cov}} \equiv c_{\text{cov}}(\eta,\overline{n}_B),
\qquad
R_{\text{ach}} \equiv R_{\text{ach}}(\eta,\overline{n}_B).
\]

Define, for any random variable \(X\), the strict-outage distribution function
\[
F_X^{<}(x)\triangleq \mathbb{P}[X<x],
\]
and the corresponding strict-outage quantile
\[
Q_X^{<}(\epsilon)\triangleq \sup\{x:\,F_X^{<}(x)\le \epsilon\}.
\]

\emph{Step 1: Convert the probabilistic constraints into deterministic bounds.}

For the covertness constraint,
\[
\mathbb{P}\!\left[q > \frac{2\delta}{\sqrt{n}}\,c_{\text{cov}}\right]
=
\mathbb{P}\!\left[c_{\text{cov}} < \frac{q\sqrt{n}}{2\delta}\right]
=
F_{c_{\text{cov}}}^{<}\!\left(\frac{q\sqrt{n}}{2\delta}\right).
\]
Therefore,
\[
\mathbb{P}\!\left[q > \frac{2\delta}{\sqrt{n}}\,c_{\text{cov}}\right]\le \epsilon_{\text{cov}}
\quad\Longleftrightarrow\quad
F_{c_{\text{cov}}}^{<}\!\left(\frac{q\sqrt{n}}{2\delta}\right)\le \epsilon_{\text{cov}},
\]
which is equivalent to
\[
\frac{q\sqrt{n}}{2\delta}\le Q_{c_{\text{cov}}}^{<}(\epsilon_{\text{cov}}).
\]
Hence,
\[
q \le \frac{2\delta}{\sqrt{n}}\,Q_{c_{\text{cov}}}^{<}(\epsilon_{\text{cov}}).
\]
Combining with the box constraint \(q\le 1\) gives
\[
q \le \min\!\left\{1,\;\frac{2\delta}{\sqrt{n}}\,Q_{c_{\text{cov}}}^{<}(\epsilon_{\text{cov}})\right\}
\eqqcolon q_{\max}.
\]

For the reliability constraint,
\[
\mathbb{P}[R>R_{\text{ach}}]
=
\mathbb{P}[R_{\text{ach}}<R]
=
F_{R_{\text{ach}}}^{<}(R).
\]
Therefore,
\[
\mathbb{P}[R>R_{\text{ach}}]\le \epsilon_{\text{rel}}
\quad\Longleftrightarrow\quad
F_{R_{\text{ach}}}^{<}(R)\le \epsilon_{\text{rel}},
\]
which is equivalent to
\[
R \le Q_{R_{\text{ach}}}^{<}(\epsilon_{\text{rel}})
\eqqcolon R_{\max}.
\]

\emph{Step 2: Maximize over the induced feasible rectangle.}

The feasible set induced by the probabilistic constraints and the box constraints is
\[
[0,q_{\max}]\times[0,R_{\max}].
\]
If \(R_{\max}=0\), the rectangle collapses to the zero-rate edge and the optimal throughput is \(T^*=0\), so no nontrivial covert link is feasible under the specified budgets. Otherwise, since \(T(q,R)=qR\) is nondecreasing in both arguments for \(q\ge 0\) and \(R\ge 0\), its maximum over this rectangle is attained at the upper-right corner \((q_{\max},R_{\max})\). Therefore,
\[
(q^*,R^*)=(q_{\max},R_{\max}),
\qquad
T^*=q_{\max}R_{\max}.
\]

This proves Theorem~\ref{thm:optimal_throughput}. \hfill\(\square\)

\medskip

To obtain Corollary~\ref{cor:continuous_quantile_reduction}, note that when the relevant outage thresholds lie at continuity points of the corresponding CDFs, we have \(F_X^{<}(x)=F_X(x)\) at those points, and the strict-outage quantile \(Q_X^{<}(\epsilon)\) reduces to the usual left-continuous quantile \(F_X^{-1}(\epsilon)\). Substituting this into the theorem yields \eqref{eq:q_max} and \eqref{eq:R_max}. \hfill\(\square\)

\section{Proof of Theorem~\ref{thm:foc_risk_adjusted} (First-order conditions for differentiable interior optima)}
\label{Proof:thm_foc_risk_adjusted}
Under the continuity assumptions required here, the strict-outage probabilities
\[
\mathbb{P}[c_{\text{cov}}<x]
\quad\text{and}\quad
\mathbb{P}[R_{\text{ach}}<r]
\]
coincide with the ordinary CDF evaluations
\[
F_{c_{\text{cov}}}(x)
\quad\text{and}\quad
F_{R_{\text{ach}}}(r),
\]
respectively.

Define
\[
J(q,R)\;\triangleq\; qR
-\lambda_{\text{cov}}\,\mathbb{P}\!\left[q > \tfrac{2\delta}{\sqrt{n}}\,c_{\text{cov}}\right]
-\lambda_{\text{rel}}\,\mathbb{P}\!\left[R > R_{\text{ach}}\right],
\]
over the box $\mathcal{D}=\{(q,R): 0\le q\le 1,\; 0\le R\le 1\}$, where $\lambda_{\text{cov}},\lambda_{\text{rel}}\ge 0$.
Assuming continuity and differentiability at the evaluation points, we can write
\[
J(q,R) \;=\; qR
-\lambda_{\text{cov}}\,F_{c_{\text{cov}}}\!\left(\tfrac{q\sqrt{n}}{2\delta}\right)
-\lambda_{\text{rel}}\,F_{R_{\text{ach}}}(R).
\]
Assume that $F_{c_{\text{cov}}}$ and $F_{R_{\text{ach}}}$ are differentiable at the evaluation points (the non-differentiable case can be handled via subgradients). If $(q^*,R^*)$ is an \emph{interior} maximizer with $0<q^*<1$ and $0<R^*<1$, the first-order optimality conditions reduce to stationarity:
\begin{align*}
\frac{\partial J}{\partial q}(q^*,R^*)
&= R^* - \lambda_{\text{cov}}\,
\frac{\sqrt{n}}{2\delta}\,
f_{c_{\text{cov}}}\!\left(\tfrac{q^*\sqrt{n}}{2\delta}\right)=0,\\
\frac{\partial J}{\partial R}(q^*,R^*)
&= q^* - \lambda_{\text{rel}}\,
f_{R_{\text{ach}}}(R^*)=0,
\end{align*}
where $f_{c_{\text{cov}}}$ and $f_{R_{\text{ach}}}$ are the PDFs (derivatives of the CDFs) at the corresponding points. Rearranging yields
\[
R^* \;=\; \lambda_{\text{cov}}\,
\frac{\sqrt{n}}{2\delta}\,
f_{c_{\text{cov}}}\!\left(\tfrac{q^* \sqrt{n}}{2\delta}\right),
\qquad
q^* \;=\; \lambda_{\text{rel}}\,
f_{R_{\text{ach}}}(R^*),
\]
which are the stated first-order conditions \eqref{eq:foc_q}–\eqref{eq:foc_R}.

If no interior stationary point exists (e.g., due to non-differentiability or the equations having no solution in $(0,1)\times(0,1)$), the maximizer lies on the boundary of $\mathcal{D}$, where at least one of $q$ or $R$ is at a bound, including the possibility of the trivial zero-rate point $R=0$. In that case, the problem reduces to one-dimensional maximization along the active boundary and can be checked directly. In the numerical results, we therefore evaluate the risk-adjusted objective directly on a dense grid over $(q,R)\in[0,1]\times[0,1]$ to avoid artifacts caused by non-smooth empirical CDFs and boundary-pinned solutions. The first-order conditions, \eqref{eq:foc_q} and \eqref{eq:foc_R}, remain useful for structural interpretation whenever an interior differentiable optimum exists.

\section{Proof of Lemma 1}
\label{app:proof_lemma_1}

We provide the detailed derivation for the closed-form quantile functions presented in Lemma~\ref{lem:exponential_case}.

\subsection*{Derivation of the Covertness Quantile}
Let $X = \overline{n}_B \sim \mathrm{Exp}(\lambda)$ and define
\[
c_{\text{cov}}(\eta_0,X) = k \sqrt{X + \eta_0 X^2},
\qquad
k = \frac{\sqrt{2\eta_0}}{1-\eta_0}.
\]
This is a monotonically increasing function of $X$ for $X \ge 0$. Let $F_{c_{\text{cov}}}(c)$ denote the CDF of the induced random variable $c_{\text{cov}}(\eta_0,X)$, i.e.,
\[
F_{c_{\text{cov}}}(c) = \mathbb{P}\!\big(c_{\text{cov}}(\eta_0,X)\le c\big).
\]
\begin{align*}
\mathbb{P}\!\big(k \sqrt{X + \eta_0 X^2} \le c\big)
&= \mathbb{P}\!\left(X + \eta_0 X^2 \le \frac{c^2}{k^2}\right) \\
&= \mathbb{P}\!\left(\eta_0 X^2 + X - \frac{c^2}{k^2} \le 0\right).
\end{align*}
The positive root of the quadratic $\eta_0 x^2 + x - (c/k)^2 = 0$ is
\[
x_{\text{root}} = \frac{-1 + \sqrt{1 + 4\eta_0(c/k)^2}}{2\eta_0},
\]
so the inequality holds for $0 \le X \le x_{\text{root}}$. Therefore,
\[
F_{c_{\text{cov}}}(c)
=
\mathbb{P}(X \le x_{\text{root}})
=
1 - e^{-\lambda x_{\text{root}}}.
\]
To find the quantile $F^{-1}_{c_{\text{cov}}}(\epsilon)$, we set
\[
F_{c_{\text{cov}}}(c) = \epsilon
\]
and solve for $c$.
\begin{align*}
    \epsilon &= 1 - \exp\left(-\lambda \frac{-1 + \sqrt{1 + 4\eta_0(c/k)^2}}{2\eta_0}\right) \\
    \ln(1-\epsilon) &= -\frac{\lambda}{2\eta_0}\left(-1 + \sqrt{1 + 4\eta_0(c/k)^2}\right) \\
    1 - \frac{2\eta_0}{\lambda}\ln(1-\epsilon) &= \sqrt{1 + 4\eta_0(c/k)^2}
\end{align*}
Let $Z = 1 - \frac{2\eta_0}{\lambda}\ln(1-\epsilon)$. Squaring both sides and solving for $c$ yields:
\[
c^2 = \frac{k^2(Z^2 - 1)}{4\eta_0} \implies c = k \sqrt{\frac{Z^2-1}{4\eta_0}}.
\]
This is the closed-form expression for $F^{-1}_{c_{\text{cov}}}(\epsilon_{\text{cov}})$. Therefore, after enforcing the box constraint $q\le 1$, the optimal transmission probability is
\[
q_{\max}
=
\min\!\left\{1,\; \frac{2\delta}{\sqrt{n}}\,k\sqrt{\frac{Z^2-1}{4\eta_0}}\right\}.
\]

\paragraph{Derivation of the Reliability Quantile.}
Let \(R_{\text{ach}} = R_{\text{ach}}(\eta_0,X)\), where \(X=\overline{n}_B\sim \mathrm{Exp}(\lambda)\). Since \(R_{\text{ach}}(\eta_0,X)\) is monotonically decreasing in \(X\), its \(\epsilon_{\text{rel}}\)-quantile is obtained by evaluating \(R_{\text{ach}}\) at the \((1-\epsilon_{\text{rel}})\)-quantile of the noise distribution. Specifically, if
\[
x_{\epsilon} \;=\; F_X^{-1}(1-\epsilon_{\text{rel}})
\;=\; -\frac{1}{\lambda}\ln(\epsilon_{\text{rel}}),
\]
then
\[
F^{-1}_{R_{\text{ach}}}(\epsilon_{\text{rel}})
\;=\;
R_{\text{ach}}(\eta_0,x_{\epsilon})
\;=\;
\Big[\,1 - H\!\big(\vec p(\eta_0,-\tfrac{1}{\lambda}\ln(\epsilon_{\text{rel}}))\big)\,\Big]^+.
\]
Hence,
\[
R_{\max}
\;=\;
\Big[\,1 - H\!\big(\vec p(\eta_0,-\tfrac{1}{\lambda}\ln(\epsilon_{\text{rel}}))\big)\,\Big]^+.
\]
\end{document}